%% file: main.tex
\documentclass{lmcs} 
\pdfoutput=1

\usepackage[utf8]{inputenc}

\usepackage{lastpage}
\lmcsdoi{19}{4}{31}
\lmcsheading{}{\pageref{LastPage}}{}{}%
{Sep.~08,~2022}{Dec.~18,~2023}{}

\keywords{$\lambda$-calculus, enriched category theory, quantale, categorical logic}

\usepackage{hyperref}
\usepackage{proof}
\usepackage[all]{xy}
\usepackage{tikz-cd}
\usepackage{stmaryrd}
\usepackage{mathrsfs}
\usepackage{bbm}
\usepackage{mathtools}
\usepackage{tabularx}
\usepackage{xspace}
\usepackage{hyperref}
\usepackage{cleveref}
\usepackage{caption}
\usepackage{amssymb}
\usepackage[nomargin,footnote]{fixme}
\usepackage{multirow}
\fxsetup{status=draft}
\crefname{section}{\S\hspace{-3pt}}{\S\hspace{-3pt}}
\crefname{figure}{Fig.}{Figs.}
\crefname{defi}{Def.}{Defs.}
\crefname{thm}{Thm.}{Thms.}
\crefname{lem}{Lemma}{Lemmata}
\crefname{exa}{Ex.}{Ex.}
\input{macros}
\theoremstyle{plain} 
\crefname{corollary}{Corollary}{Corollaries}
\theoremstyle{plain}\newtheorem{proposition}[thm]{Proposition} 
\crefname{proposition}{Prop.}{Props.}

\newcommand{\CPTP}{\catfont{CPTP}}
\newcommand{\Iso}{\catfont{Iso}}
\newcommand{\Ban}{\catfont{Ban}}

\newcommand{\normal}{\prog{normal}}

\newcommand{\unit}{\prog{unit}}
\newcommand{\real}{\prog{real}}
\newcommand{\meas}{\mathcal{M}}
\newcommand{\Tr}{\mathrm{Tr}}

\begin{document}

\title[Syntactic side categories enriched over generalised
metric spaces]{The syntactic side of autonomous categories enriched over generalised
metric spaces}

\titlecomment{{\lsuper*}This paper is an extended version
        of~\cite{dahlqvist22}.  It includes proofs omitted in the \emph{op.
        cit.} and new examples.  It also includes new technical results: among
        other things an equivalence theorem and an extension (from the linear
        setting to the affine one) of the results established in the~\emph{op.
cit.} }

\author[F.~Dahlqvist]{Fredrik Dahlqvist\lmcsorcid{0000-0003-2555-0490}}[a]
\author[R.~Neves]{Renato Neves\lmcsorcid{0000-0002-8787-2551}}[b]

\address{Queen Mary University of London and University College London}	
\email{f.dahlqvist@qmul.ac.uk}  

\address{University of Minho and INESC-TEC}	
\email{nevrenato@di.uminho.pt}  





\begin{abstract}
  Programs with a continuous state space or that interact with physical
  processes often require notions of equivalence going beyond the standard
  binary setting in which equivalence either holds or does not hold. In this
  paper we explore the idea of equivalence taking values in a quantale
  $\mathcal{V}$, which covers the cases of (in)equations and (ultra)metric
  equations among others.

  Our main result is the introduction of a \emph{$\mathcal{V}$-equational
  deductive system} for linear $\lambda$-calculus together with a proof that it
  is sound and complete. In fact we go further than this, by showing that
  linear $\lambda$-theories based on this $\V$-equational system form a
  category equivalent to a category of autonomous categories enriched over
  `generalised metric spaces'. If we instantiate this result to inequations, we
  get an equivalence with autonomous categories enriched over partial orders.
  In the case of (ultra)metric equations, we get an equivalence with autonomous
  categories enriched over (ultra)metric spaces. Additionally, we show that
  this syntax-semantics correspondence extends to the affine setting.
  
We use our results to develop examples of inequational and metric equational
systems for higher-order programming in the setting of real-time,
probabilistic, and quantum computing.  \end{abstract}

\maketitle

\section{Introduction}\label{S:one}
Programs frequently act over a \emph{continuous} state space or
interact with physical processes like time progression or the movement
of a vehicle. Such features naturally call for notions of
approximation and refinement integrated in different aspects of
program equivalence. Our paper falls in this line of research.
Specifically, our aim is to integrate notions of approximation and
refinement into the \emph{equational system} of linear
$\lambda$-calculus~\cite{benton92,mackie93,maietti05}.

The core idea that we explore in this paper is to have equations
$t =_q s$ labelled by elements $q$ of a quantale $\mathcal{V}$. This
covers a wide range of situations, among which the cases of
(in)equations~\cite{kurz2017quasivarieties,adamek20} and metric
equations~\cite{mardare16,mardare17}. The latter case is perhaps less
known: it consists of equations $t =_{\epsilon}s$ labelled by a
non-negative rational number $\epsilon$ which represents the `maximum
distance' that the two terms $t$ and $s$ can be from each other. In
order to illustrate metric equations, consider a programming language
with a (ground) type $X$ and a signature of operations
$\Sigma = \{ \prog{wait_n} : X \to X \mid n \in \Nats \}$ that model
time progression over computations of type $X$. Specifically,
$\prog{wait_n}(x)$ reads as ``add a latency of $n$ seconds to the
computation $x$''. In this context the following metric equations arise naturally:
\begin{flalign}\label{ax}
  \prog{wait_0}(x) =_0 x \hspace{1cm} \prog{wait_n}(\prog{wait_m}(x)) =_0
  \prog{wait_{n + m }}(x) \hspace{1cm}
  \infer{\prog{wait_n}(x) =_\epsilon \prog{wait_m}(x)}{\epsilon = |m - n|}
\end{flalign}
Equations $t =_0 s$ state that the terms $t$ and $s$ are exactly the same and
equations $t =_\epsilon s$ state that $t$ and $s$ differ by \emph{at most}
$\epsilon$ seconds in their execution time.

\noindent
\textbf{Contributions.}  In this paper we introduce an equational deductive
system for linear $\lambda$-calculus in which equations are labelled by
elements of a quantale $\mathcal{V}$.  By using key features of a quantale's
structure, we show that this deductive system is \emph{sound and complete} for
a class of enriched symmetric monoidal closed categories (\ie enriched
\emph{autonomous} categories).  In particular, if we fix $\mathcal{V}$ to be
the Boolean quantale this class of categories consists of autonomous categories
enriched over partial orders. If we fix $\mathcal{V}$ to be the (ultra)metric
quantale, then this class of categories consists of autonomous categories
enriched over (ultra)metric spaces. The aforementioned example of wait calls
fits in the setting in which $\mathcal{V}$ is the metric quantale.  Our result
provides this example with a sound and complete metric equational system, where
the models are all those autonomous categories enriched over metric spaces that
can soundly interpret the axioms of wait calls~\eqref{ax}.

The next contribution of our paper falls in a major topic of categorical logic:
to establish a \emph{syntax-semantics bidirectional correspondence} between logical systems
and certain classes of categories (in a
nutshell, this allows to translate categorical assertions or constructions
into logical ones and \emph{vice-versa}). A famous result of this kind is the
correspondence between $\lambda$-calculus and Cartesian closed categories,
formalised in terms of an equivalence between categories respective to both
structures. Intuitively, the equivalence states that $\lambda$-theories are the
syntactic counterpart  of Cartesian closed categories and that the latter are
the semantic counterpart of the former~\cite{lambek88,crole93}. An analogous
result is known to exist for linear $\lambda$-calculus and autonomous
categories~\cite{mackie93,maietti05}.  Here we extend the latter result to the
setting of $\V$-equations (\ie equations labelled by elements of a quantale
$\mathcal{V}$). Specifically, we prove the existence of an equivalence between
linear $\lambda$-theories based on $\V$-equations and autonomous categories
enriched over `generalised metric spaces'. 

\noindent \textbf{Outline.} Section~\ref{sec:back} recalls linear
$\lambda$-calculus and its equational system together with corresponding proofs
of soundness, completeness, and the aforementioned equivalence with a category
of autonomous categories.  The contents of this section are adaptations of
results presented in~\cite{benton92,mackie93,crole93,maietti05}, the main
difference being that we forbid the exchange rule to be explicitly part of
linear $\lambda$-calculus (instead it is only admissible). This choice is
important to ensure that judgements in the calculus have \emph{unique}
derivations, which allows to refer to their interpretations
unambiguously~\cite{shulman19}. Section~\ref{sec:main} presents the main
contributions of this paper. It walks a path analogous to
Section~\ref{sec:back}, but now in the setting of $\mathcal{V}$-equations. As
we will see, the semantic counterpart of moving from equations to
$\mathcal{V}$-equations is to move from ordinary categories to categories
enriched over \emph{$\mathcal{V}$-categories}. The latter, often regarded as
generalised metric spaces, are central entities in a fruitful area of enriched
category theory that aims to treat uniformly different kinds of `structured
sets', such as partial orders, fuzzy partial orders, and (ultra)metric
spaces~\cite{lawvere73,stubbe14,velebil19}. Our results are applicable to all
these cases. Section~\ref{sec:examples} presents some examples of
$\mathcal{V}$-equational axioms (and corresponding models) for three
computational paradigms, namely real-time, probabilistic, and quantum computing
(in all these paradigms notions of approximation take a central role).
Specifically, in Section~\ref{sec:examples} we will revisit the axioms of wait
calls~\eqref{ax} and consider an inequational variant. Then we will study a
metric axiom for probabilistic programs and show that the category of Banach
spaces and short linear maps is a model for the resulting metric theory. Next
we turn our attention to quantum computing and introduce a metric axiom that
reflects the fact that implementations of quantum operations can only
approximate the intended behaviour.  We build a corresponding model over a
certain category of presheaves based on the concept of a quantum channel.  We
also illustrate  how our deductive system allows to compute an approximate
distance between two probabilistic/quantum programs easily as opposed to
computing an exact distance `semantically' which tends to involve quite complex
operators. 

Finally Section~\ref{sec:concl} provides two concluding notes: first a proof
that our results extend from the linear to the affine case. Second the
presentation of a functorial connection (in terms of adjunctions) between our
results and previous (algebraic) semantics of linear logic~\cite{paiva99,
maietti05}. Section~\ref{sec:concl} then ends with a brief exposition of future
work.  We assume knowledge of $\lambda$-calculus and category
theory~\cite{mackie93,maietti05,lambek88,maclane98}. 

\noindent
\textbf{Related work.}  Several approaches to incorporating quantitative
information in programming languages have been explored in the literature.
Closest to this work are various approaches targeted at $\lambda$-calculi.
In~\cite{crubille2015metric,crubille2017metric} a notion of distance called
\emph{context distance} is developed, first for an affine, then for a more
general $\lambda$-calculus, with probabilistic programs as the main motivation.
\cite{gavazzo18} considers a notion of quantale-valued \emph{applicative
(bi)similarity}, an operational \emph{coinductive} technique used for showing
contextual equivalence between two programs.  Recently, \cite{pistone21}
presented several Cartesian closed categories of generalised metric spaces that
provide a quantitative semantics to simply-typed $\lambda$-calculus based on a
generalisation of logical relations.  None of these examples reason about
distances in a quantitative equational system, and in this respect our work is
closer to the metric universal algebra developed in \cite{mardare16,mardare17}.

A different approach consists in encoding quantitative information via a type
system.  In particular, graded (modal) types
\cite{girard1992bounded,gaboardi2016combining,orchard2019quantitative} have
found applications in \eg differential privacy~\cite{reed10} and information
flow \cite{abadi1999core}.  This approach is to some extent orthogonal to ours
as it mainly aims to model coeffects, whilst we aim to reason about the
intrinsic quantitative nature of $\lambda$-terms acting \eg on continuous or
ordered spaces.

Quantum programs provide an interesting example of intrinsically quantitative
programs, by which we mean that the metric structure on quantum states is not
seen as arising from (co)effects.  Recently \cite{hung19} showed how the issue
of noise in a quantum while-language can be handled by developing a
\emph{deductive system} to determine how similar a quantum program is from its
idealised, noise-free version; an approach very much in the spirit of this
work.

\section{Background: linear $\lambda$-theories and autonomous categories}
\label{sec:back}

In this section we recall linear $\lambda$-calculus, which can be regarded as a
term assignment system for the exponential free, multiplicative fragment of
\emph{intuitionistic linear logic}. We briefly recall that it is sound and
complete w.r.t.\ autonomous categories, the reader will find more details
in~\cite{mackie93,benton92,maietti05}. We then present categories of linear
$\lambda$-theories and of autonomous categories, and show that they are
equivalent. This second result is shown in~\cite{maietti05}, but we take the
more general approach of allowing functors to preserve autonomous structures
only \emph{up-to isomorphism} (\ie in the style of~\cite{crole93} which studies
the Cartesian variant). For that reason we will present it in some detail.
Finally note that both results are particular instances of the general results
presented in the ensuing section when one takes $\mathcal{V}$ to be the Boolean
quantale.

\subsection{Linear $\lambda$-calculus, soundness and completeness}
 \label{sec:folin}
 Let us start by fixing a \emph{class} $G$ of ground types.  The grammar of
 types for linear $\lambda$-calculus is given by:
\[
  \typeA ::=  X \in G \mid \typeI \mid 
  \typeA \otimes \typeA \mid \typeA \multimap \typeA
\] 
We also fix a \emph{class} $\Sigma$ of sorted operation symbols $f :
\typeA_1,\dots,\typeA_n \to \typeA$ with $n \geq 1$.  As usual, we use Greek
letters $\Gamma,\Delta, E,\dots$ to denote \emph{typing contexts}, \ie lists
$x_1 : \typeA_1,\dots,x_n : \typeA_n$ of typed variables such that each
variable $x_i, 1\leq i\leq n,$ occurs at most once in $x_1,\dots,x_n$.

We use the notion of \emph{shuffle} to build a linear typing system such
that the \emph{exchange rule} is admissible and each judgement $\Gamma \vljud v
: \typeA$ has a \emph{unique derivation} -- this will allow us to refer to a
judgement's denotation $\sem{\Gamma \vljud v : \typeA}$ unambiguously.  By
shuffle we mean a permutation of typed variables in a context sequence
$\Gamma_1,\dots,\Gamma_n$ such that for all $i \leq n$ the relative order of
the variables in $\Gamma_i$ is preserved~\cite{shulman19}. For example, if
$\Gamma_1 = x : \typeA, y : \typeB$ and $\Gamma_2 = z : \typeC$ then $z :
\typeC, x : \typeA, y : \typeB$ is a shuffle of $\Gamma_1,\Gamma_2$ but $y :
\typeB, x : \typeA, z : \typeC$ is \emph{not}, because we changed the order in
which $x$ and $y$ appear in $\Gamma_1$.  As explained in~\cite{shulman19},
such a restriction on relative orders is
crucial for judgements having unique derivations.  We denote by
$\Shuff(\Gamma_1;\dots;\Gamma_n)$ the set of shuffles on
$\Gamma_1,\dots,\Gamma_n$.

The term formation rules of the linear $\lambda$-calculus are shown in
Fig.~\ref{fig:lang}. They correspond to the natural deduction rules of the
exponential-free, multiplicative fragment of intuitionistic linear logic.
Substitution is defined as expected, yielding a particularly well-behaved
calculus.

\begin{figure*}[h!]
\small{
    \begin{tabular}{|llllc|}
    \hline
    & & & & \\
    \multicolumn{4}{|l}{
      \infer[\rulename{ax}]{E
        \vljud f(v_1,\dots,v_n) : \typeA}
      {\Gamma_i \vljud v_i : \typeA_i
      \quad f : \typeA_1,\dots,\typeA_n \to \typeA \in \Sigma
      \quad E \in \Shuff(\Gamma_1;\dots;\Gamma_n) }
      }
      &
      \infer[\rulename{hyp}]{x : \typeA \vljud x : \typeA}{}
      \\
      & & & & \\
      \infer[\rulename{\typeI_i}]{- \vljud \ast : \typeI}{}
      &
      \multicolumn{4}{r|}{
      \infer[\rulename{\otimes_e}]{E \vljud \prog{pm}\ v\ \prog{to}\
      x \otimes y.\ w : \typeC}
      {\Gamma \vljud v : \typeA \otimes \typeB
      \quad \Delta , x : \typeA, y : \typeB \vljud w : \typeC
      \quad E \in \Shuff(\Gamma;\Delta)}
       }
	\\
     & & & & \\
      \multicolumn{3}{|l}{
      \infer[\rulename{\otimes_i}]{E \vljud v \otimes w : \typeA \otimes
        \typeB}{\Gamma \vljud v : \typeA \quad \Delta \vljud w : \typeB
      \quad E \in \Shuff(\Gamma;\Delta)}
      }
      &
      \multicolumn{2}{r|}{
      \infer[\rulename{\typeI_e}]{E  \vljud v \ \prog{to}\ \ast.\ w
      : \typeA}
      {\Gamma \vljud v : \typeI \quad \Delta 
        \vljud w : \typeA \quad E \in \Shuff(\Gamma;\Delta)}
      }
      \\
      & & & & \\
      \multicolumn{3}{|l}{
      \infer[\rulename{\multimap_i}]{\Gamma \vljud \lambda x : \typeA . \,
      v : \typeA \multimap \typeB}
      {\Gamma, x : \typeA \vljud v : \typeB}
 
      }
      &
      \multicolumn{2}{r|}{
      \infer[\rulename{\multimap_e}]{E \vljud v \, w : \typeB}
      {\Gamma \vljud  v : \typeA \multimap \typeB \quad
      \Delta \vljud  w : \typeA \quad E \in \Shuff(\Gamma;\Delta)}      
     }
      \\
      \hline
    \end{tabular}
    }
  \caption{Term formation rules of linear $\lambda$-calculus.}
  \label{fig:lang}
\end{figure*}

\begin{thm}\label{thm:properties}
        The calculus defined by the rules of \cref{fig:lang} enjoys the
        following properties:
	\begin{enumerate}
                \item (Unique typing) For any two judgements $\Gamma \vljud v :
                        \typeA$ and $\Gamma \vljud v : \typeA'$, we have
                        $\typeA = \typeA'$; \label{i:unique_type}
                \item (Unique derivation) Every judgement $\Gamma \vljud v :
                        \typeA$ has a unique derivation; \label{i:unique_der}
                \item (Exchange) For every judgement $\Gamma, x : \typeA, y :
                        \typeB, \Delta \vljud v : \typeC$ \label{i:exch} we can
                        derive $\Gamma, y : \typeB, x : \typeA, \Delta \vljud v
                        : \typeC$; 
                \item (Substitution) For all judgements $\Gamma, x : \typeA
                        \vljud v : \typeB$ and $\Delta \vljud w : \typeA$ we
                        can derive $\Gamma,\Delta \vljud v[w/x] : \typeB$.
                        \label{i:subst} 
	\end{enumerate}
\end{thm}

We now recall the interpretation of judgements $\Gamma \vljud v: \typeA$ in a
symmetric monoidal closed (\ie autonomous) category $\catC$. We start by fixing
some notation. For all $\catC$-objects $X,Y,Z$, $\sw_{X,Y} : X \otimes Y \to Y
\otimes X$ denotes the symmetry morphism, $\lambda_X : \typeI \otimes X \to X$
the left unitor,  $\rho_X: X\otimes \typeI\to X$ the right unitor, and
$\alpha_{X,Y,Z} : X \otimes (Y \otimes Z) \to (X \otimes Y) \otimes Z$ the left
associator. Moreover for all $\catC$-morphisms $f : X \otimes Y \to Z$ the
morphism $\overline{f} : X \to (Y \multimap Z)$ denotes the corresponding
curried version (right transpose).

For all ground types $X \in G$ we postulate an interpretation $\sem{X}$ as a
$\catC$-object.  Types are interpreted  inductively using the unit $\typeI$,
the tensor $\otimes$ and the internal hom $\multimap$ of autonomous categories.
Given a non-empty context $\Gamma = \Gamma', x : \typeA$, its interpretation is
defined by $\sem{\Gamma', x : \typeA} = \sem{\Gamma'} \otimes \sem{\typeA}$ if
$\Gamma'$ is non-empty and $\sem{\Gamma', x : \typeA} = \sem{\typeA}$
otherwise. The empty context $-$ is interpreted as $\sem{-} = \typeI$.  Given
$X_1,\dots,X_n \in \catC$ we write $X_1 \otimes \dots \otimes X_n$ for the
$n$-tensor $(\dots (X_1 \otimes X_2) \otimes \dots ) \otimes X_n$, and
similarly for $\catC$-morphisms. We will often omit subscripts in
components of natural transformations if no ambiguities arise.

We will also need `housekeeping' morphisms to handle interactions between
context interpretation and the autonomous structure of $\catC$. Given
$\Gamma_1,\dots,\Gamma_n$ we denote by $\spl_{\Gamma_1; \dots;\Gamma_n} :
\sem{\Gamma_1, \dots, \Gamma_n} \to \sem{\Gamma_1} \otimes \dots \otimes
\sem{\Gamma_n}$ the morphism that splits $\sem{\Gamma_1, \dots, \Gamma_n}$ into
$\sem{\Gamma_1} \otimes \dots \otimes \sem{\Gamma_n}$ which is defined as
follows. Given $\Gamma_1$ and $\Gamma_2$, $\spl_{\Gamma_1 ; \Gamma_2} :
\sem{\Gamma_1, \Gamma_2} \to \sem{\Gamma_1} \otimes \sem{\Gamma_2}$ is defined
by
        \[
                \spl_{-;\Gamma} = \lambda^{-1} \hspace{1.3cm}
                \spl_{\Gamma; -} = \rho^{-1} \hspace{1.3cm}
                \spl_{\Gamma; x : \typeA} = \id \hspace{1.3cm}
                \spl_{\Gamma; (\Delta, x : \typeA) } =
                \alpha^{-1} \comp  (\spl_{\Gamma; \Delta} \otimes \id)
        \]
For $n > 2$, $\spl_{\Gamma_1; \dots ; \Gamma_n} : \sem{\Gamma_1, \dots,
\Gamma_n} \to \sem{\Gamma_1} \otimes \dots \otimes \sem{\Gamma_n}$ is defined
via the previous definition and induction on the size of $n$:
        \[
                \spl_{\Gamma_1; \dots ; \Gamma_n} = 
                (\spl_{\Gamma_1; \dots ; \Gamma_{n -1}} \otimes \id) \comp
                \spl_{\Gamma_1, \dots, \Gamma_{n-1}; \Gamma_n}
        \]
We denote by $\join_{\Gamma_1;\dots;\Gamma_n}$ the
        inverse of $\spl_{\Gamma_1;\dots;\Gamma_n}$.  Next, given
        $\Gamma, x : \typeA, y : \typeB, \Delta$ we denote by $\exch_{\Gamma,
        \underline{x : \typeA, y : \typeB}, \Delta} : \sem{\Gamma, x : \typeA,
y : \typeB, \Delta} \to \sem{\Gamma, y : \typeB, x : \typeA, \Delta}$ the
morphism permuting $x$ and $y$:
        \[ 
                \exch_{\Gamma; \underline{x : \typeA, y : \typeB}, \Delta}
                = \join_{\Gamma; y: \typeB, x : \typeA; \Delta} \comp\,
                (\id \otimes \sw \otimes \id) \comp 
                \spl_{\Gamma; x : \typeA, y : \typeB; \Delta}
        \]
The shuffling morphism $\sh_E : \sem{E} \to \sem{\Gamma_1,\dots,\Gamma_n}$ is
defined as a suitable composition of exchange morphisms.  Whenever convenient
we will drop variables or even the whole subscript in the housekeeping
morphisms.

\begin{figure*}[h!]
        \scalebox{0.81}{
\begin{tabular}{|llllr|}
	\hline
	& & & & \\
	\multicolumn{4}{|l}
	{
	\infer[]{\sem{E \vljud f(v_1,\dots,v_n) : \typeA} = \sem{f} \comp
      (m_1 \otimes \dots \otimes m_n) \comp \spl_{\Gamma_1;\dots;\Gamma_n}
      \comp \sh_E}
      {\sem{\Gamma_i \vljud v_i : \typeA_i} = m_i
      \quad f : \typeA_1,\dots,\typeA_n \to \typeA \in \Sigma
      \quad E \in \Shuff(\Gamma_1 \dots \Gamma_n) }
	}
	&
	\infer[]{\sem{x : \typeA \vljud x : \typeA} = \id_{\sem{\typeA}}}{}
	\\
	& & & & \\
	 \infer[]{\sem{- \vljud \ast : \typeI} = \id_{\sem{\typeI}}}{}      
	 &
	 \multicolumn{4}{r|}
	 {
	 \infer[]{\sem{E \vljud \prog{pm}\ v\ \prog{to}\ x \otimes y.\
          w : \typeC} =
        n \comp \join_{\Delta; \typeA; \typeB}
        \comp \, \alpha \comp \sw \comp (m \otimes \id) \comp
        \spl_{\Gamma;\Delta} \comp \sh_E}
      {\sem{\Gamma \vljud v : \typeA \otimes \typeB} = m
      \quad \sem{\Delta , x : \typeA, y : \typeB \vljud w : \typeC} = n
      \quad E \in \Shuff(\Gamma;\Delta)}
	 }
	 \\
	 & & & & \\
	 \multicolumn{2}{|l}
	 {
	 \infer[]{\sem{E \vljud v \otimes w : \typeA \otimes \typeB} =
        (m \otimes n) \comp \spl_{\Gamma;\Delta}
        \comp \sh_E}{\sem{\Gamma \vljud v : \typeA} = m
      \quad \sem{\Delta \vljud w : \typeB} = n
      \quad E \in \Shuff(\Gamma;\Delta)} 
	 }
	 &
	\multicolumn{3}{c|}
	{
	\infer[]{\sem{E  \vljud v \ \prog{to}\ \ast.\ w : \typeA} =
        n \comp \lambda \comp
      (m \otimes \id) \comp \spl_{\Gamma; \Delta} \comp \sh_E}
      {\sem{\Gamma \vljud v : \typeI} = m \quad \sem{\Delta 
          \vljud w : \typeA} = n \quad E \in \Shuff(\Gamma;\Delta)}
	}	 
	 \\
        & & & & \\
        \multicolumn{2}{|l}
	{
        \infer[]{\sem{\Gamma \vljud \lambda x : \typeA . \, v : \typeA \multimap \typeB} =
        \overline{ (m \comp \join_{\Gamma; \typeA}) }}
        {\sem{\Gamma, x : \typeA \vljud v : \typeB} = m }
        } &
	\multicolumn{3}{c|}
	{
        \infer[]{\sem{E \vljud v \, w : \typeB} = \app \comp (m \otimes n)
        \comp \spl_{\Gamma;\Delta} \comp \sh_E}
        {\sem{\Gamma \vljud  v : \typeA \multimap \typeB} = m \quad
        \sem{\Delta \vljud  w : \typeA} = n \quad E \in \Shuff(\Gamma;\Delta)}     
	}
        \\	
        \hline
\end{tabular}
}
  \caption{Judgement interpretation on an autonomous category $\catC$.}
  \label{fig:lang_sem}
\end{figure*}

For every operation symbol $f : \typeA_1, \dots, \typeA_n \to \typeA$
in $\Sigma$ we postulate an interpretation
$\sem{f} : \sem{\typeA_1} \otimes \dots \otimes \sem{\typeA_n} \to
\sem{\typeA}$ as a $\catC$-morphism. The interpretation of judgements
is defined by induction over
derivations according to the rules in \cref{fig:lang_sem}.

As detailed in~\cite{benton92,mackie93,maietti05}, linear $\lambda$-calculus
comes equipped with a class of equations, given in ~\cref{fig:eqs}, specifically
\emph{equations-in-context} $\Gamma \vljud v = w : \typeA$, that corresponds to
the axiomatics of autonomous categories.  As usual, we omit the context and
typing information of the equations in Fig.~\ref{fig:eqs}, which can be
reconstructed in the usual way.  

\begin{figure*}[h!]
\captionsetup{width=\linewidth}
\small{
	\begin{tabular}{| c | c |}
		\hline
		Monoidal structure
		&
		Higher-order structure
		\\
		\hline
		\begin{tabular}{r c l}
	    $\prog{pm}\ v \otimes w\ \prog{to}\ x \otimes y.\ u$ &
	    $=$ & $u[v/x,w/y]$ \\
	    $\prog{pm}\ v\ \prog{to}\ x \otimes y.\
	    u[x \otimes y / z]$ &
	    $=$ & $u[v/z]$ \\
	    $\ast\ \prog{to}\ \ast.\ v$ & $=$ & $v$ \\
	    $v\ \prog{to}\ \ast.\ w[\ast/ z]$ & $=$ & $w[v/z]$
	  	\end{tabular}
  		&
	  	\begin{tabular}{r c l}
                      $(\lambda x : \typeA .\ v)\ w$ &$=$& $v[w/x]$ \\
                      $\lambda x : \typeA . (v\ x)$ &$=$& $v$
                \end{tabular}
	  	\\
                \hline
                \multicolumn{2}{|c|}{Commuting conversions} 
                \\
                \hline
                \multicolumn{2}{|c|}{
                \begin{tabular}{r c l}
                        $u[v\ \prog{to} \ast.\ w/z]$ &$=$& $v\ \prog{to}\ \ast.\ u[w/z]$ \\
	                $u[\prog{pm}\ v\ \prog{to}\ x \otimes y.\ w/z]$
	                &$=$& $\prog{pm}\ v\ \prog{to}\ x \otimes y.\ u[w/z]$
                \end{tabular}
                }\\
	  	\hline
	\end{tabular}
	\caption{Equations corresponding to the axiomatics of autonomous categories.}
  \label{fig:eqs}
}
\end{figure*}

The next step is to prove that the equational
schema listed in~\cref{fig:eqs} is sound w.r.t. autonomous categories. For
that effect, we will use the following \emph{exchange and substitution
lemma} which can be proved straightforwardly by induction using the coherence theorem of symmetric monoidal categories~\cite{maclane98}.

\begin{lem}[Exchange and Substitution]
  \label{lem:exch_subst_inter}
  For any judgements
  $\Gamma, x : \typeA, y : \typeB, \Delta \vljud v : \typeC$,
  $\Gamma, x : \typeA \vljud v : \typeB$, and
  $\Delta \vljud w : \typeA$, the following equations hold in
  every autonomous category $\catC$:
  \begin{align*}
    \sem{\Gamma, x : \typeA, y : \typeB, \Delta \vljud v : \typeC} & =
    \sem{\Gamma, y : \typeB, x : \typeA, \Delta \vljud v : \typeC}
    \comp \exch_{\Gamma, \underline{\typeA,\typeB}, \Delta} \\
    \sem{\Gamma, \Delta \vljud v[w / x] : \typeB} & =
    \sem{\Gamma, x : \typeA \vljud v : \typeB} \comp
    \join_{\Gamma;\typeA} \comp\, (\id \otimes \sem{\Delta \vljud w : \typeA})
    \comp \spl_{\Gamma;\Delta}
  \end{align*}
\end{lem}

\begin{thm}
  \label{theo:bsound}
  The equations presented in Fig.~\ref{fig:eqs} are sound w.r.t.\  judgement
  interpretation. More specifically if $\Gamma \vljud v = w : \typeA$ is one of
  the equations in Fig.~\ref{fig:eqs} then $\sem{\Gamma \vljud v : \typeA} =
  \sem{\Gamma \vljud w : \typeA}$.
\end{thm}

\begin{proof}[Proof sketch]
  Follows from Lemma~\ref{lem:exch_subst_inter},
  the coherence theorem for symmetric monoidal categories, and
  naturality. We exemplify this with one of the commuting conversions.
    \begin{flalign*}
    & \, \sem{\Gamma, \Delta,E \vljud u[v\ \prog{to}\ \ast.\ w/x]
    : \typeB} & \\
    & = \sem{u}  \comp
    \join_{\Gamma;\typeA} \comp\, (\id \otimes\, \sem{v\ \prog{to}\ \ast.\ w})
    \comp \spl_{\Gamma; \Delta,E} & \text{\{Lemma~\ref{lem:exch_subst_inter}\}}\\
    & \defeq \sem{u} \comp \join_{\Gamma;\typeA} \comp\, (\id \otimes\,
    (\sem{w} \comp \lambda \comp (\sem{v} \otimes \id) \comp
    \spl_{\Delta;E})) \comp \spl_{\Gamma;\Delta,E} \\
    & = \sem{u} \comp \join_{\Gamma;\typeA} \comp\, (\id \otimes\, \sem{w})
    \comp (\id \otimes\,
    (\lambda \comp (\sem{v} \otimes \id) \comp
            \spl_{\Delta;E})) \comp \spl_{\Gamma;\Delta,E} & \\
    & = \sem{u} \comp \join_{\Gamma;\typeA} \comp\, (\id \otimes\, \sem{w})
    \comp \spl_{\Gamma;E} \comp \join_{\Gamma;E} \comp\, (\id \otimes\,
    (\lambda \comp (\sem{v} \otimes \id) \comp
            \spl_{\Delta;E})) \comp \spl_{\Gamma;\Delta,E} & \\
    & = \sem{u[w/x]} \comp \join_{\Gamma;E} \comp\, (\id \otimes\,
    (\lambda \comp (\sem{v} \otimes \id) \comp\,
            \spl_{\Delta;E})) \comp \spl_{\Gamma;\Delta,E}& 
            \text{\{Lemma~\ref{lem:exch_subst_inter}\}}\\
    & = \sem{u[w/x]} \comp \lambda \comp (\sem{v} \otimes \id) \comp
            \spl_{\Delta;\Gamma,E} \comp \sh_{\Gamma,\Delta,E} & 
            \text{\{c\}}\\
    & \defeq \sem{\Gamma,\Delta,E \vljud v\ \prog{to}\ \ast.\ u[w/x] :
      \typeB}
    &  \qedhere
  \end{flalign*}
\end{proof}

\begin{defi}[Linear $\lambda$-theories]\label{defn:theory}
  Consider a tuple $(G,\Sigma)$ consisting of a class $G$ of ground types and a
  class $\Sigma$ of sorted operation symbols.  A \emph{linear $\lambda$-theory}
  $((G,\Sigma),Ax)$ is a triple such that $Ax$ is a class of
  equations-in-context over linear $\lambda$-terms built from $(G,\Sigma)$.
\end{defi}
The elements of $Ax$ are called the \emph{axioms} of the theory. Let $Th(Ax)$
be the smallest congruence that contains $Ax$, the equations listed in
\cref{fig:eqs}, and that is closed under exchange and substitution
(\cref{thm:properties}).  We call the elements of $Th(Ax)$ the
\emph{theorems} of the theory.
  
\begin{defi}[Models of linear $\lambda$-theories]\label{defn:model}
        Consider a linear $\lambda$-theory $((G,\Sigma),Ax)$ and also an autonomous
        category $\catC$. Suppose that for each $X \in G$ we have an
        interpretation $\sem{X}$ that is a $\catC$-object and analogously for
        the operation symbols. This interpretation structure is a \emph{model}
        of the theory if all axioms are satisfied by the interpretation.
\end{defi}

\begin{thm}[Soundness \& Completeness]~\label{theo:sound_compl}
  Consider a linear $\lambda$-theory $\mathscr{T}$. Then an equation $\Gamma
  \vljud v = w : \typeA$ is a theorem of $\mathscr{T}$ iff it is satisfied by
  all models of the theory.
\end{thm}

\begin{proof}[Proof sketch]
        Soundness follows by induction over the rules that define $Th(Ax)$ and
        by \cref{theo:bsound}.  Completeness is based on the idea of a
        \emph{Lindenbaum-Tarski algebra}: it follows from building the
        \emph{syntactic category} $\Syn(\mathscr{T})$ of $\mathscr{T}$ (also
        known as term model), showing that it possesses an autonomous structure
        and also that equality $\sem{\Gamma \vljud v : \typeA} = \sem{\Gamma
        \vljud w : \typeA}$ in the syntactic category is equivalent to
        provability $\Gamma \vljud v = w : \typeA$ in the theory.  The
        syntactic category of $\mathscr{T}$ has as objects the types of
        $\mathscr{T}$ and as morphisms $\typeA \to \typeB$ the equivalence
        classes (w.r.t.\  provability) of terms $v$ for which we can derive $x :
        \typeA \vljud v : \typeB$.
\end{proof}

\subsection{Equivalence theorem between linear $\lambda$-theories and
autonomous categories}
\label{sec:eq}

We now
present a category of $\lambda$-theories, a category
of autonomous categories, and then show that they are equivalent. 
In order to prepare the stage, we will start by establishing a
bijective correspondence (up-to isomorphism) between models of a
$\lambda$-theory $\mathscr{T}$ on a category $\catC$ and \emph{autonomous}
functors $\Syn(\mathscr{T}) \to \catC$. Although not strictly necessary for
establishing the aforementioned equivalence, this bijection has multiple
benefits: first it will allow us to formally see models as functors and thus
opens up the possibility of applying functorial constructions to them. For
example, the notion of a natural isomorphism (between functors) carries to the
notion of a model isomorphism. Second, it will later on help us motivate the
notion of a morphism and equivalence between $\lambda$-theories. Third, it will
help us establish the aforementioned equivalence of categories. Our proof
of the bijective correspondence between models and autonomous functors is
inspired by an analogous one~\cite{crole93} for the Cartesian case.

We first recall the definition of an autonomous functor. We will use $I_\catC$
to denote the unit of a monoidal category $\catC$, and drop the subscript
whenever no ambiguities arise.

\begin{defi}
        \label{def:autf}
        A functor $\funF : \catC \to \catD$ between two monoidal categories
        $\catC$ and $\catD$ is called \emph{monoidal} if it is equipped
        with a morphism $\un : I_\catD \to \funF (I_\catC)$ and a natural
        transformation $\mu_{X,Y} : \funF X \otimes_\catD \funF Y \to \funF (X
        \otimes_\catC Y)$ such that the following diagrams commute:
        \begin{center}
        \begin{tabular}{l r}
        \xymatrix@R=16pt{
        	\funF X \otimes I \ar[r]^(0.46){\id \otimes \un} 
        	\ar[d]_{\rho} &
        	\funF X \otimes \funF I \ar[d]^{\mu_{}} \\
        	\funF X & \ar[l]^(0.55){\funF \rho} \funF(X \otimes I)
        } 
    	&
    	\qquad
    	\multirow{4}{*}[-14pt]{ 
    		\xymatrix@R=18pt@C=66pt{
    			(\funF X \otimes \funF Y) \otimes \funF Z \ar[r]^{\alpha^{-1}} 
    			\ar[d]_{\mu\, \otimes\, \id} &
    			\funF X \otimes (\funF Y \otimes \funF Z) 
    			\ar[d]^{\id \otimes\, \mu} \\
    			\funF ( X \otimes Y) \otimes \funF Z 
    			\ar[d]_{\mu} & 
    			\funF X \otimes \funF ( Y \otimes Z) 
    			\ar[d]^{\mu} \\
    			\funF ((X \otimes Y) \otimes Z) \ar[r]_{\funF \alpha^{-1}} & 
    			\funF (X \otimes (Y \otimes Z))
    		} 
    	}
    	\\ 
		\\
    	 \xymatrix@R=16pt{
    		I \otimes \funF X \ar[r]^(0.46){\un \otimes \id} 
    		\ar[d]_{\lambda} &
    		\funF I \otimes \funF X \ar[d]^{\mu_{}} \\
    		\funF X & \ar[l]^(0.55){\funF \lambda} \funF(I \otimes X)
    	}       
        \end{tabular}
    	\end{center}
    
        \noindent The functor $\funF : \catC \to \catD$ is \emph{strong monoidal} if $\un
        : I_\catD \to \funF(I_\catC)$ is an isomorphism and if $\mu_{X,Y} :
        \funF X \otimes_\catD \funF Y \to \funF(X \otimes_\catC Y)$ is a
        natural isomorphism. We call $\funF$ \emph{strict} if $\un$ is the
        identity and $\mu_{X,Y}$ is also the identity for all $\catC$-objects
        $X,Y$. Assume next that both categories $\catC$ and $\catD$ are
        symmetric monoidal. Then we say that $\funF$ is \emph{symmetric
        monoidal} if it is monoidal and moreover the diagram below commutes.
        \[
                \xymatrix@R=16pt{
                        \funF X \otimes \funF Y \ar[r]^{\sw} \ar[d]_{\mu} & 
                        \funF Y \otimes \funF X \ar[d]^{\mu}
                        \\
                        \funF (X \otimes Y) \ar[r]_{\funF \sw}
                        & \funF(Y \otimes X)
                }
        \]
        Finally consider the following $\catD$-morphism:
        \[
               \xymatrix{
                       \funF(X \multimap Y) \otimes \funF X 
                       \ar[r]^{\mu} & 
                       \funF((X \multimap Y) \otimes X)  \ar[r]^(0.68){\funF \app} & 
                       \funF Y
               }
        \]
        We call $\funF$ \emph{autonomous} if it is symmetric strong monoidal
        and if the right transpose of the previous morphism is an isomorphism.
        We call $\funF$ \emph{strict autonomous} if it is symmetric strict 
        monoidal and if the right transpose is the identity.
\end{defi}

Next, for a linear $\lambda$-theory $\mathscr{T}$ we will show that autonomous
functors $\funF : \catC \to \catD$ send models of $\mathscr{T}$ on $\catC$ to
models of $\mathscr{T}$ on $\catD$.  In other words, the autonomous functor
$\funF$ permits a `$\mathscr{T}$-respecting' change of interpretation domain.
This mapping will be useful for proving the aforementioned bijection between
autonomous functors and models.  Let us consider an interpretation $\sem{-}_M$
of a linear $\lambda$-theory $\mathscr{T}$ over an autonomous category $\catC$,
and an autonomous functor $\funF : \catC \to \catD$. We define the
interpretation $\sem{-}_{F_\ast M}$ on ground types by
\[
\sem{X}_{\funF_\ast M} \defeq \funF \sem{X}_M, \qquad X\in G.
\]
In order to define the interpretation $\sem{f}_{\funF_\ast M}$ of
operation symbols $f: \typeA_1,\dots,\typeA_n \to \typeA$, we start by inductively  building an isomorphism $h_\typeA : \sem{\typeA}_{\funF_\ast M} \to
\funF \sem{\typeA}_M$ for all types $\typeA$. 
\begin{itemize}
	\item If $\typeA=X\in G$ is a ground type then
	\[
	h_\typeA\defeq \id_{F\sem{\typeA}_M}:\sem{\typeA}_{\funF_\ast M} \to  \funF \sem{\typeA}_M
	\]
	\item  If $\typeA=\typeI$ then
	\[h_\typeI \defeq \un :
	\sem{\typeI}_{\funF_\ast M} = I_\catD \to \funF I_\catC = \funF
	\sem{\typeI}_M
	\]
	\item  If $\typeA=\typeA_1 \otimes \typeA_2$ then
	\[
	h_{\typeA_1 \otimes \typeA_2}\defeq\mu \comp
	(h_{\typeA_1} \otimes h_{\typeA_2}): \sem{\typeA_1 \otimes \typeA_2}_{\funF_\ast M} \to	\funF\sem{\typeA_1 \otimes \typeA_2}_M
	\]
	\item If $\typeA=\typeA_1 \multimap \typeA_2$ then
	\[
	h_{\typeA_1 \multimap \typeA_2}\defeq m^{-1} \comp (-
	\comp h_{\typeA_1}^{-1}) \comp (h_{\typeA_2} \comp -): \sem{\typeA_1 \multimap \typeA_2}_{\funF_\ast M}\to \funF\sem{\typeA_1
		\multimap\typeA_2}_M
	\]
	where $m:
	\funF(\sem{\typeA_1}_M \multimap \sem{\typeA_2}_M) \to \funF \sem{\typeA_1}_M
	\multimap \funF \sem{\typeA_2}_M$ is the right transpose of $\funF (\app)
	\comp \mu$
\end{itemize}
With these isomorphisms in place, we define the interpretation of operation symbols as:
\[
        \sem{f}_{\funF_\ast M} \defeq h^{-1}_{\typeA} \comp \funF \sem{f}_M 
        \comp h_{\typeA_1 \otimes \dots \otimes \typeA_n}
\]
Note that if $\funF$ is strict the isomorphism $h_\typeA$ collapses to the identity and $\sem{f}_{\funF_\ast M}=F\sem{f}_M$.
Whenever no ambiguities arise we will drop the subscript in $h_\typeA$.  For a
context $\Gamma$ let us define $\sem{\Gamma}^\funF_M$ via the equations
$\sem{-}^\funF_M = I_\catD$, $\sem{\Gamma, x  : \typeA}^\funF_M =
\sem{\Gamma}^\funF_M \otimes \funF \sem{\typeA}_M$ if $\Gamma$ is non-empty and
$\sem{\Gamma, x : \typeA}^\funF_M = \funF \sem{\typeA}_M$ otherwise.
Intuitively, $\sem{\Gamma}^\funF_M$ corresponds to a pointwise application of
$\funF$ to the tuple of objects in the tensor $\sem{\Gamma}_M$. We can then
extend the isomorphism $h: \sem{\typeA}_{\funF_\ast M} \to \funF
\sem{\typeA}_M$ to contexts $h[\Gamma] : \sem{\Gamma}_{\funF_\ast M} \to
\sem{\Gamma}^\funF_M$ via induction.

In order to show that $\funF_\ast M$ is a model of $\mathscr{T}$ whenever $M$
is, we need to fix extra notation.  Recall the natural transformation
$\mu_{X,Y} : \funF X \otimes \funF Y \to \funF (X \otimes Y)$.  Then for $n
\geq 0$, let us define the morphism $\mu[X_1,\dots,X_n] : \funF X_1\, \otimes
\dots \otimes \funF X_n \to \funF(X_1 \otimes \dots \otimes X_n)$ by setting
$\mu[-] = \un : I_\catD \to \funF (I_\catC)$, $\mu[X] = \id$ and
$\mu[X_1,\dots,X_n] = \mu \comp (\mu[X_1,\dots,X_{n-1}] \otimes \id)$ for $n
\geq 2$. Given a context $\Gamma = x_1 : \typeA_1, \dots, x_n : \typeA_n$ we
use $\mu[\Gamma]$ to denote the morphism
$\mu[\sem{\typeA_1}_M,\dots,\sem{\typeA_n}_M] : \sem{\Gamma}^\funF_M \to \funF
\sem{\Gamma}_M$.  Finally, note that for $n$ contexts $\Gamma_1,\dots,\Gamma_n$
we can build a `split' morphism $\spl^\funF_{\Gamma_1;\dots;\Gamma_n} :
\sem{\Gamma_1,\dots,\Gamma_n}^\funF_M \to \sem{\Gamma_1}^\funF_M \otimes \dots
\otimes \sem{\Gamma_n}^\funF_M$ analogous to the morphism
$\spl_{\Gamma_1;\dots;\Gamma_n} : \sem{\Gamma_1,\dots,\Gamma_n}_M \to
\sem{\Gamma_1}_M \otimes \dots \otimes \sem{\Gamma_n}_M$.  We will also need
the following lemma which states that the morphisms $h_\typeA$ and the monoidal
structure of $\funF : \catC \to \catD$ commute with the housekeeping morphisms
used for judgement interpretation. The proof is straightforward, using the
definitions and properties of symmetric monoidal categories.

\begin{lem}\label{lem:housekeeping}
	The following diagrams commute:
		\begin{equation}
        \label{dia_h_sp}
        \xymatrix@C=60pt{
                \sem{\Gamma_1,\dots,\Gamma_n}_{\funF_\ast M} 
                \ar[r]^(0.45){\spl_{\Gamma_1;\dots;\Gamma_n}} 
                \ar[d]_{h[\Gamma_1, \dots, \Gamma_n]} &
                \sem{\Gamma_1}_{\funF_\ast M} \otimes\, \dots\, \otimes\,
                \sem{\Gamma_n}_{\funF_\ast M}
                \ar[d]^{h[\Gamma_1]\, \otimes\, \dots\, \otimes\, h[\Gamma_n]} 
                \\
                \sem{\Gamma_1,\dots,\Gamma_n}^\funF_M
                \ar[r]_(0.45){\spl^\funF_{\Gamma_1;\dots;\Gamma_n}} &
                \sem{\Gamma_1}^\funF_M \otimes\, \dots \otimes 
                \sem{\Gamma_n}^\funF_M
        }
		\end{equation}
		\begin{equation}
        \label{dia_mu_sp}
        \xymatrix@C=60pt{
                \sem{\Gamma_1,\dots,\Gamma_n}^\funF_{M} 
                \ar[r]^(0.47){\spl^\funF_{\Gamma_1;\dots;\Gamma_n}} 
                \ar[dd]_{\mu[\Gamma_1,\dots,\Gamma_n]} &
                \sem{\Gamma_1}^\funF_{M} \otimes\, \dots \otimes
                \sem{\Gamma_n}^\funF_{M}
                \ar[d]^{\mu[\Gamma_1]\, \otimes\, \dots\, \otimes\,
                \mu[\Gamma_n]} 
                \\ 
                & \funF \sem{\Gamma_1}_M \otimes \dots \otimes
                \funF \sem{\Gamma_n}_M
                \ar[d]^{\mu[\sem{\Gamma_1}_M,\dots,
                \sem{\Gamma_n}_M]}\\
                \funF \sem{\Gamma_1,\dots,\Gamma_n}_M
                \ar[r]_(0.47){\funF \spl_{\Gamma_1;\dots;\Gamma_n}} &
                \funF (\sem{\Gamma_1}_M \otimes\, \dots \otimes\, 
                \sem{\Gamma_n}_M)
        }
		\end{equation}
		\begin{minipage}{0.49\textwidth}
		\begin{equation}
        \label{dia_h_sh}
        \xymatrix@R=16pt@C=36pt{
                \sem{E}_{\funF_\ast M}
                \ar[r]^(0.42){\sh_{E}}
                \ar[d]_{h[E]} &
                \sem{\Gamma_1,\dots,\Gamma_n}_{\funF_\ast M} 
                \ar[d]^{h[\Gamma_1,\dots,\Gamma_n]} \\
                \sem{E}^\funF_{M}
                \ar[r]_(0.42){\sh^\funF_{E}} 
                &
                \sem{\Gamma_1,\dots,\Gamma_n}^\funF_{M} 
        } 
		\end{equation}
		\end{minipage}
		\begin{minipage}{0.49\textwidth}
		\begin{equation}
        \label{dia_mu_sh}
        \xymatrix@R=16pt@C=36pt{
                \sem{E}^\funF_{M}
                \ar[r]^(0.42){\sh^\funF_{E}}
                \ar[d]_{\mu[E]} &
                \sem{\Gamma_1,\dots,\Gamma_n}^\funF_{M} 
                \ar[d]^{\mu[\Gamma_1,\dots,\Gamma_n]} \\
                \funF \sem{E}_{M}
                \ar[r]_(0.42){\funF \sh_{E}} 
                &
                \funF \sem{\Gamma_1,\dots,\Gamma_n}_{M} 
        }
		\end{equation}
		\end{minipage}
\end{lem}

\begin{proposition}
        \label{prop:model}
        If $M$ is a model of the $\lambda$-theory $\mathscr{T}$, then $\funF_\ast M$ is also model of $\mathscr{T}$.
\end{proposition}
\begin{proof}
        Let us first assume that
        for all judgements $\Gamma \vljud v : \typeA$ the following
        equation holds.
        \[
                \sem{\Gamma \vljud v : \typeA}_{\funF_\ast M} = 
                h^{-1}_\typeA \comp 
                \funF \sem{\Gamma \vljud v : \typeA}_{M} \comp \mu[\Gamma] \comp h[\Gamma]
        \]
        Diagrammatically this corresponds to the commutativity of the diagram,
        \[
                \xymatrix@C=60pt{
                        \sem{\Gamma}_{\funF_\ast M} \ar[r]^{\sem{v}_{\funF_\ast M}}
                        \ar[d]_{\mu[\Gamma]\ \comp\ h[\Gamma]}
                        & \sem{\typeA}_{\funF_\ast M} \ar@/^/[d]^{h} \\
                        \funF \sem{\Gamma}_M \ar[r]_{\funF \sem{v}_{M}}
                        & \ar@/^/[u]^{h^{-1}} \funF \sem{\typeA}_M 
                }
        \]
        From this assumption we reason as follows:
        \begin{flalign*}
                & v = w & \\
                & \Rightarrow \sem{v}_M = \sem{w}_M & \\
                & \Rightarrow \funF \sem{v}_M = \funF \sem{w}_M & \\
                & \Rightarrow h^{-1} \comp \funF \sem{v}_M = h^{-1} \comp \funF \sem{w}_M &  \\
                        & \Rightarrow h^{-1} \comp \funF \sem{v}_M \comp \mu[\Gamma] 
                        \comp h[\Gamma] = h^{-1} \comp 
                        \funF \sem{w}_M \comp \mu[\Gamma] \comp h[\Gamma]  
                        & \\
                        & \Rightarrow 
                        \sem{v}_{\funF_\ast M} = \sem{w}_{\funF_\ast M}  
                        & \text{\{Assumption\}}
        \end{flalign*}
        which indeed entails our claim. The rest of the proof amounts to
        showing that our assumption holds. This follows from induction over the
        semantic rules of linear $\lambda$-calculus. We present a selection of
        some cases, the other ones are obtained in an analogous manner.
          \begin{flalign*}
                & \, \sem{v \otimes w}_{\funF_\ast M} \\
                & = (\sem{v}_{\funF_\ast M} \otimes \sem{w}_{\funF_\ast M} ) \comp
                \spl_{\Gamma ; \Delta} \comp \sh_E \\
                & = ((h^{-1} \comp \funF \sem{v} \comp \mu[\Gamma] \comp h[\Gamma]) \otimes
                (h^{-1} \comp \funF \sem{w} \comp \mu[\Delta] \comp h[\Delta])) 
                \comp \spl_{\Gamma; \Delta} \comp \sh_E \\
                & = h^{-1} \comp \mu \comp ((\funF \sem{v} \comp \mu[\Gamma] \comp h[\Gamma]) 
                \otimes
                (\funF \sem{w} \comp \mu[\Delta] \comp h[\Delta])) 
                \comp \spl_{\Gamma; \Delta} \comp \sh_E &
                \text{\{Def. $h$ on $\typeA \otimes \typeB$\}}
                \\
                & =  h^{-1} \comp \mu \comp (\funF \sem{v} \otimes \funF \sem{w}) \comp
                ((\mu[\Gamma] \comp h[\Gamma]) \otimes (\mu[\Delta] \comp h[\Delta]))
                \comp \spl_{\Gamma;\Delta} \comp \sh_E 
                \\
                & = h^{-1} \comp \funF (\sem{v} \otimes \sem{w}) \comp \mu \comp 
                ((\mu[\Gamma] \comp h[\Gamma]) \otimes (\mu[\Delta] \comp h[\Delta]))
                \comp \spl_{\Gamma;\Delta} \comp \sh_E 
                & \text{\{Naturality of $\mu$\}}
                \\
                & = h^{-1} \comp \funF (\sem{v} \otimes \sem{w}) \comp \mu \comp 
                (\mu[\Gamma] \otimes \mu[\Delta]) \comp (h[\Gamma] \otimes h[\Delta])
                \comp \spl_{\Gamma;\Delta} \comp \sh_E 
                \\
                & = h^{-1} \comp \funF (\sem{v} \otimes \sem{w}) \comp 
                \funF \spl_{\Gamma;\Delta} \comp \funF \sh_E \comp \mu[E] \comp h[E]
                & \{\star\}
                \\
                & = h^{-1} \comp \funF \sem{v \otimes w} \comp \mu[E] \comp h[E]
                & 
        \end{flalign*}
        where $\{\star\}$ follows from combining the commutative diagrams
        of ~\cref{lem:housekeeping}. 
        
        \begin{flalign*}
                & \, \sem{\lambda x : \typeA.\ v }_{\funF_\ast M} \\
                & = \overline{\sem{v}_{\funF_\ast M} \comp
                \join_{\Gamma; x : \typeA}} \\
                & = \overline{h^{-1} \comp \funF \sem{v} 
                \comp \mu[\Gamma, x : \typeA] \comp h [\Gamma, x : \typeA] \comp
                \join_{\Gamma; x : \typeA}} \\
                & = \overline{ h^{-1} \comp \funF \sem{v} 
                \comp \funF \join_{\Gamma; x : \typeA} 
                \comp \, \mu \comp ((\mu[\Gamma] \comp h[\Gamma]) \otimes h))}
                & \{\dagger\}
                \\
                & = \overline{h^{-1} \comp \funF \sem{v} 
                \comp \funF \join_{\Gamma; x : \typeA} 
                \comp \, \mu \comp (\id \otimes h)}
                \comp \mu[\Gamma] \comp h[\Gamma]
                & \left \{\overline{(f \comp (g \otimes \id)} = \overline{ f} \comp g \right \}
                \\
                & = (h^{-1} \comp -) \comp \overline{ 
                \funF \sem{v} 
                \comp \funF \join_{\Gamma; x : \typeA} 
                \comp \, \mu \comp (\id \otimes h)}
                \comp \mu[\Gamma] \comp h[\Gamma]
                & \left \{ \overline{g \comp f} = (g \comp -) \comp \overline{f} \right \}
                \\
                & = (h^{-1} \comp -) \comp (- \comp h) \comp \overline{
                \funF \sem{v} 
                \comp \funF \join_{\Gamma; x : \typeA} 
                \comp \mu}
                \comp \mu[\Gamma] \comp h[\Gamma]
                & \left \{ \overline {g \comp (\id \otimes f)} 
                = (- \comp f) \comp \overline{g} \right \}
                \\
                & = (h^{-1} \comp -) \comp (- \comp h) \comp \overline{
                \funF (\sem{v} 
                \comp \join_{\Gamma; x : \typeA} )
                \comp \mu}
                \comp \mu[\Gamma] \comp h[\Gamma]
                & 
                \\
                & = (h^{-1} \comp -) \comp (- \comp h) \comp m \comp \funF
                (\overline{\sem{v} 
                \comp \join_{\Gamma; x : \typeA} )}
                \comp \mu[\Gamma] \comp h[\Gamma]
                & \left \{ \overline{\funF f \comp \mu} = m \comp \funF \overline{f} \right \}
                \\
                & = h^{-1} \comp  
                \funF (\overline {\sem{v} 
                \comp \join_{\Gamma; x : \typeA} )}
                \comp \mu[\Gamma] \comp h[\Gamma]
                &\text{\{Def. $h$ on $\typeA \multimap \typeB$\}}
                \\
                & = h^{-1} \comp \funF \sem{\lambda x: \typeA.\ v }
                \comp \mu[\Gamma] \comp h[\Gamma]
                \end{flalign*}
            where $\{\dagger\}$ follows from $\{\star\}$ above combined with $\join_{\Gamma_1;\dots;\Gamma_n}$ being the
            inverse of $\spl_{\Gamma_1;\dots;\Gamma_n}$.
        \end{proof}

Finally, we use the previous result to formally establish the bijective
correspondence (up-to isomorphism) that was discussed in the beginning of the
current subsection.

\begin{thm} \label{th:str}
        Let $\mathscr{T}$ be a linear $\lambda$-theory. Every autonomous functor
        $\funF : \Syn(\mathscr{T}) \to \catC$ induces a model ${\funF_\ast
        \Syn(\mathscr{T})}$ of $\mathscr{T}$ and every model $M$ of
        $\mathscr{T}$ induces a strict autonomous functor $([v] \mapsto
        \sem{v}_M) : \Syn(\mathscr{T}) \to \catC$. Furthermore these
        constructions are inverse to each other up-to isomorphism, in the sense
        that,
        \begin{align*}
                ([v] \mapsto \sem{v}_{\funF_\ast \Syn(\mathscr{T})}) 
                \cong \funF \hspace{1cm} 
                \text{ and } \hspace{1cm}
                ([v] \mapsto \sem{v}_M)_\ast \Syn(\mathscr{T}) = M
        \end{align*}
\end{thm}

\begin{proof}
       Let  us first focus on the mapping that sends functors to models.
       Consider an autonomous functor $\funF : \Syn(\mathscr{T}) \to \catC$.
       Then observe that $\Syn(\mathscr{T})$ corresponds to a model of
       $\mathscr{T}$ and thus by Proposition~\ref{prop:model} we conclude that
       $\funF_\ast \Syn(\mathscr{T})$ must be a model of $\mathscr{T}$. For the
       inverse direction, we start with a model $M$ and build the functor
       $\Syn(\mathscr{T}) \to \catC$ that sends the equivalence class $[v]$
       into $\sem{v}_M$. It is straightforward to
       prove that this last functor is strict autonomous.

       Our next step is to prove the existence of a natural isomorphism $([v]
       \mapsto \sem{v}_{\funF_\ast \Syn(\mathscr{T})}) \cong \funF$. For that
       effect observe that for all types $\typeA$ we already have,
       \[
               \sem{\typeA}_{\funF_\ast \Syn(\mathscr{T})} \stackrel{h}{\cong} \funF
               \sem{\typeA}_{\Syn(\mathscr{T})} = \funF \typeA
       \]
       Observe as well that the corresponding naturality square,
       \[
                \xymatrix@C=45pt{
                \sem{\typeA}_{\funF_\ast \Syn(\mathscr{T})} 
                \ar[d]_{\sem{v}_{\funF_\ast \Syn(\mathscr{T})}}
                \ar@/^/[r]^{h}
                & \funF \sem{\typeA}_{\Syn(\mathscr{T})}
                \ar@/^/[l]^{h^{-1}} 
                \ar[d]^{\funF \sem{v}_{\Syn(\mathscr{T})}\, =\, \funF [v]}  \\
                \sem{\typeB}_{\funF_\ast \Syn(\mathscr{T})} 
                \ar@/^/[r]^{h}
                & \funF \sem{\typeB}_{\Syn(\mathscr{T})}
                \ar@/^/[l]^{h^{-1}}
                }
       \]
       is an instance of equation $\sem{v}_{\funF_\ast M} = h^{-1} \comp \funF
       \sem{v}_{M} \comp \mu[\Gamma] \comp h[\Gamma]$ (specifically, with
       $|\Gamma| = 1$ and $M = \Syn(\mathscr{T})$) which was proved in
       Proposition~\ref{prop:model}. Finally, we show that the equation $([v]
       \mapsto \sem{v}_M)_\ast \Syn(\mathscr{T}) = M$ holds. First,
       \begin{align*}
              \sem{G}_{([v] \mapsto \sem{v}_M)_\ast \Syn(\mathscr{T})} 
              = ([v] \mapsto \sem{v}_M) (\sem{G}_{\Syn(\mathscr{T})}) 
                = ([v] \mapsto \sem{v}_M) (G)
                = \sem{G}_M &
       \end{align*}
       and second, by keeping in mind that $[v] \mapsto \sem{v}_M$ is strict autonomous,
       \begin{flalign*}
              & \, \sem{f}_{([v] \mapsto \sem{v}_M)_\ast \Syn(\mathscr{T})} & \\ 
              & = ([v] \mapsto \sem{v}_M) (\sem{f}_{\Syn(\mathscr{T})}) & \\
              & = ([v] \mapsto \sem{v}_M) ([\prog{pm}\ x\ \prog{to}\ x_1 
               \otimes \dots \otimes x_n.\ f(x_1,\dots,x_n)]) & \\
              & = \sem{\prog{pm}\ x\ \prog{to}\ x_1 
       \otimes \dots \otimes x_n.\ f(x_1,\dots,x_n)}_M & \\
              & = \sem{f}_M &
              \qedhere
       \end{flalign*}
\end{proof}
We are now in the right setting to present the equivalence between linear
$\lambda$-theories and autonomous categories that was discussed in the paper's
introduction and beginning of Section~\ref{sec:back}. So let $\Aut$ be the category of
locally small autonomous categories and autonomous functors. Consider as well
the category $\Aut_{/_{\cong}}$ whose objects are locally small autonomous
categories and morphisms are \emph{isomorphism classes} of autonomous functors
(this category is well-defined because isomorphisms form an equivalence
relation that is closed w.r.t. pre- and
post-composition~\cite[II.8]{maclane98}). We will show that the latter category
is equivalent to a certain category $\lambda\text{-} \catfont{Th}$ whose
objects are linear $\lambda$-theories:
\begin{align} \label{equ1}
        \xymatrix@C=40pt{
                \lambda\text{-} \catfont{Th}
                \ar@<5pt>[r]_(0.50){\simeq} & 
                \Aut_{/_{\cong}}
                \ar@<5pt>[l]^(0.50){}
        }
\end{align}
We need to define the notion of a morphism in the category
$\lambda\text{-} \catfont{Th}$.  Following traditions in type
theory~\cite{crole93} we set $\lambda\text{-}
\catfont{Th}(\mathscr{T}_1,\mathscr{T}_2) :=
\Aut_{/_{\cong}}(\Syn(\mathscr{T}_1),\Syn(\mathscr{T}_2))$.  In words, a
morphism $\mathscr{T}_1 \to \mathscr{T}_2$ between $\lambda$-theories
$\mathscr{T}_1$ and $\mathscr{T}_2$ is exactly an isomorphism class of
autonomous functors $\Syn(\mathscr{T}_1) \to \Syn(\mathscr{T}_2)$ -- which 
by~\cref{th:str} are in bijective correspondence (up-to isomorphism) to
models of $\mathscr{T}_1$ on the category $\Syn(\mathscr{T}_2)$. 

Note that an isomorphism $\mathscr{T}_1 \cong \mathscr{T}_2$ in
$\lambda\text{-} \catfont{Th}$ is equivalent to the corresponding syntactic
categories being equivalent, which provides a bijection,
\[
\Aut_{/_{\cong}}(\Syn(\mathscr{T}_1),\catC)
\cong
\Aut_{/_{\cong}}(\Syn(\mathscr{T}_2),\catC)
\]
natural on $\catC$. This echoes the notion of \emph{Morita equivalence} in
universal algebra~\cite{adamek06,johnstone02}, which states that two theories
are equivalent provided that the corresponding categories of varieties are
equivalent (in our case, the variety of a $\lambda$-theory $\mathscr{T}$ on
$\catC$ is simply regarded as the set
$\Aut_{/_{\cong}}(\Syn(\mathscr{T}),\catC$). 

Next we set the stage for describing the functor $\Aut_{/_{\cong}} \to
\lambda\text{-} \catfont{Th}$ that is part of the equivalence~\eqref{equ1}.  As
we will see, it corresponds to an \emph{internal language construct} \ie it
maps an autonomous category $\catC$ to a linear $\lambda$-theory $\Lang(\catC)$
that encodes the former completely. Among other things, this is useful for
seeing morphisms of a category $\catC$ equivalently as $\lambda$-terms, or more
generally for seeing $\catC$ from a set-theoretic lens without loss of
information (for a deeper discussion on the notion of internal language, see~\cite{pitts01}).

\begin{defi}[Internal language]
\label{def:lng}
An autonomous category $\catC$ induces a linear $\lambda$-theory
$\Lang(\catC)$ whose ground types $X \in G$ are the objects of $\catC$
and whose signature $\Sigma$ of operation symbols consists of all the
morphisms in $\catC$ plus certain isomorphisms that we describe
in~\eqref{eq:iso}. The axioms of $\Lang(\catC)$ are all the equations
satisfied by the obvious interpretation on $\catC$.  In order to
explicitly distinguish the autonomous structure of $\catC$ from the
type structure of $\Lang(\catC)$ let us denote the tensor of $\catC$
by $\otimes_\catC$, the unit by $I_\catC$, and the exponential by
$\multimap_\catC$. Consider then the following map on types:
\begin{flalign}\label{eq:iso}
	\hspace{-10pt}
    i(\typeI) = I_\catC \qquad i(X) = X \qquad i(\typeA \otimes
    \typeB) = i(\typeA) \otimes_\catC i(\typeB) \qquad i(\typeA
    \multimap \typeB) = i(\typeA) \multimap_\catC i(\typeB)
\end{flalign}
For each type $\typeA$ we add an isomorphism $ i_\typeA : \typeA \cong
i(\typeA)$ to the theory $\Lang(\catC)$.
\end{defi}

The following theorem is the core of the proof that establishes the
equivalence~\eqref{equ1} and formalises the fact that the $\lambda$-theory
$\Lang(\catC)$ associated to a category $\catC$ describes the latter completely
\ie it is the internal language of $\catC$. 

\begin{thm}
        \label{th:int}
        For every autonomous category $\catC$ there exists an equivalence
        $\Syn(\Lang(\catC)) \simeq \catC$ and both functors witnessing this
        equivalence are autonomous. The functor going from the left to right
        direction is additionally strict. 
\end{thm}

\begin{proof}
        By construction, we have an interpretation of $\Lang(\catC)$ in $\catC$
        which behaves as the identity for operation symbols and ground types.
        This interpretation is by definition a model of $\Lang(\catC)$ on
        $\catC$ and by~\cref{th:str} we obtain a strict autonomous functor
        $\Syn(\Lang(\catC)) \to \catC$. The functor in the opposite direction
        behaves as the identity on objects and sends a $\catC$-morphism $f$
        into $[f(x)]$. The equivalence of categories is then shown by using the
        aforementioned isomorphisms which connect the type constructors of
        $\Lang(\catC)$ with the autonomous structure of $\catC$.

        To finish the proof we need to show that the functor $\catC \to
        \Syn(\Lang(\catC))$ is autonomous (\cref{def:autf}).  Our first step is
        to establish an isomorphism $\typeI \to I_\catC$ living in the category
        $\Syn(\Lang(\catC))$ which we define as $a : \typeI \vljud
        [i_\typeI(a)] : I_\catC$.  Analogously, for all $\catC$-objects $X$ and
        $Y$ the required isomorphism $X \otimes Y \to X\, \otimes_\catC  Y$
        living in the same category is provided by $a : X \otimes Y \vljud
        [i_{X \otimes Y}(a)] : X \otimes_\catC Y$. Next we prove that the thus
        obtained transformation is natural on $X$ and $Y$. By analysing the
        corresponding diagram we see that  it corresponds to proving the
        equality $[i_{X' \otimes Y'}(b)] \comp [\prog{pm}\ a\ \prog{to}\ x
        \otimes y.\ f(x) \otimes g(y)] = [f\, \otimes_\catC g\, (x)] \comp
        [i_{X \otimes Y}(a)]$. Then note that, by the definition of syntactic
        category and by substitution, the previous equality is entailed by
        $i_{X' \otimes Y'}(\prog{pm}\ a\ \prog{to}\ x \otimes y.\ f(x) \otimes
        g(y)) = f\, \otimes_\catC g\ (i_{X \otimes Y}(a))$. According to the
        definition of $\Lang(\catC)$ (Definition~\ref{def:lng}), this last
        equality is entailed by,
        \[
                \sem{ i_{X' \otimes Y'}(\prog{pm}\ a\ \prog{to}\ x \otimes y.\
                f(x) \otimes g(y))} = \sem{ f \otimes_\catC  g\, (i_{X \otimes Y}(a))}
        \]
        on $\catC$ and the proof that the latter holds is straightforward.

        Next let us show that the left unitality diagram commutes. An inspection
        of this diagram reveals that it corresponds to the equality
        $[ \prog{pm}\ a\ \prog{to}\ i \otimes x.\ i\ \prog{to}\ \ast.\ x] =
        [\lambda_\catC (c)] \comp [i_{I_\catC \otimes X}(b)] \comp
        [ \prog{pm}\ a\ \prog{to}\ i \otimes x.\ i_\typeI(i) \otimes x]$. Analogously
        to the previous reasoning we will prove that,
        \[
                \sem{\prog{pm}\ a\ \prog{to}\ i \otimes x.\ i\ \prog{to}\ \ast.\ x}
                =
                \sem{\lambda_\catC (i_{I_\catC \otimes X}(\prog{pm}\ a\ 
                        \prog{to}\ i \otimes x.\
                i_\typeI(i) \otimes x))}
        \]
        We start by simplifying the right-hand side of equation,
        \begin{align*}
                \sem{\lambda_\catC (i_{I_\catC \otimes X}(\prog{pm}\ a\ 
                        \prog{to}\ i \otimes x.\
                i_\typeI(i) \otimes x))}
                & =
                \sem{\lambda_\catC} \comp \sem{i_{I_\catC \otimes X}} \comp
                \sem{\prog{pm}\ a\ \prog{to}\ i \otimes x.\
                i_\typeI(i) \otimes x} \\ 
                & = \sem{\lambda_\catC} \comp \id \comp \id  \\
                & = \lambda_\catC & 
        \end{align*}
        The simplified equation amounts to stating that the left unitor of
        $\Syn(\Lang(\catC))$ (the term on the left-hand side of the equation)
        is \emph{strictly} preserved by the interpretation of $\Lang(\catC)$ on
        $\catC$. This is clearly the case because the functor
        $\Syn(\Lang(\catC)) \to \catC$ corresponding to this interpretation is
        \emph{strict} autonomous (\cref{th:str}).  This reasoning is
        also applicable to the diagrams concerning right unitality, symmetry,
        and associativity.

        The final step is to prove that
        the right transpose of the composite,
        \[
                \xymatrix@C=40pt{
                        (X \multimap_\catC Y) \otimes X
                        \ar[r]^{i_{(X \multimap_{\catC} Y) \otimes X}} &
                        (X \multimap_\catC Y)\, \otimes_\catC X 
                        \ar[r]^(0.65){\app_\catC} & Y
                }
        \]
        is an isomorphism in the category $\Syn(\Lang(\catC))$. To that effect
        note that the aforementioned right transpose is given by $f : X
        \multimap_\catC Y \vljud \left [ \lambda x : X.\ \app_\catC\ \left
        (i_{X \multimap_\catC Y \otimes_\catC X}\ (f \otimes x) \right ) \right ] : X
        \multimap Y$. It is straightforward to prove that this term is interpreted as
        the identity on $\catC$. Next note the existence of the term $x : X \multimap Y \vljud
        \left [i_ {X \multimap Y}(x) \right ] : X \multimap_\catC Y$ which is given in
        Definition~\ref{def:lng} and is also interpreted as the identity on the category
        $\catC$.  Since both terms are interpreted as the identity, by Definition~\ref{def:lng}
        they are indeed inverses of each other. 
\end{proof}

Adapting the nomenclature of~\cite{linton66} to our setting, we call a theory
$\mathscr{T}$ varietal if $\Syn(\mathscr{T})$ is locally small. In the rest of
this section we will only work with varietal theories, which allows to prove
the following theorem.

\begin{thm}
        \label{th:eq}
        There exists an equivalence,
        \begin{align*}
        \xymatrix@C=40pt{
                \lambda\text{-} \catfont{Th}
                \ar@<5pt>[r]^(0.50){\Syn}_{\simeq} & 
                \Aut_{/_{\cong}}
                \ar@<5pt>[l]^(0.50){\Lang}
        }
        \end{align*}
\end{thm}
\begin{proof}
The functor $\Syn : \lambda\text{-}\catfont{Th} \to \Aut_{/_{\cong}}$ sends a
$\lambda$-theory to its syntactic category and acts as the identity on
morphisms. For the inverse direction, recall that for every autonomous category
$\catC$~\cref{th:int} provides autonomous functors $e_\catC : \catC \to
\Syn(\Lang(\catC))$ and $e'_\catC : \Syn(\Lang(\catC)) \to \catC$. Then we
define  $\Lang : \Aut_{/_{\cong}}\to \lambda\text{-}\catfont{Th}$ as the
functor that sends an autonomous category $\catC$ to its theory $\Lang(\catC)$
and that sends an isomorphism class $[\funF]$ of autonomous functors of type
$\catC \to \catD$ into $[ e_\catD \comp \funF \comp e'_\catC]$.  This last
mapping is well-defined because if $\funF \cong \funF'$ then $e_\catD \comp
\funF \comp e'_\catC \cong\, e_\catD \comp \funF' \comp e'_\catC$.
Furthermore it respects the functorial laws because $e_\catC \comp \Id \comp
e'_\catC \cong \Id$ and thus $[e_\catC \comp \Id \comp e'_\catC] = [\Id]$, also
for all autonomous functors $\funF : \catC \to \catA$ and $\funG : \catA \to
\catD$ we have the isomorphism $e_\catD \comp \funG \comp \funF \comp e'_\catC \cong e_\catD
\comp \funG \comp e'_\catA  \comp e_\catA \comp \funF \comp e'_\catC$ and
therefore $[e_\catD \comp \funG \comp \funF \comp e'_\catC] = [e_\catD \comp
\funG \comp e'_\catA  \comp e_\catA \comp \funF \comp e'_\catC]$.

The next step is to prove the existence of isomorphisms $\Syn(\Lang(\catC))
\cong \catC$ in $\Aut_{/_{\cong}}$ and $\Lang(\Syn(\mathscr{T})) \cong
\mathscr{T}$ in $\lambda\text{-}\catfont{Th}$. The former arises directly 
from~\cref{th:int} and the fact that we are now working with isomorphism
classes of autonomous functors. For the latter we appeal 
to~\cref{th:int} to establish an equivalence
$\Syn(\Lang(\Syn(\mathscr{T}))) \simeq \Syn(\mathscr{T})$; then we resort to
the previous remark on isomorphisms of $\lambda$-theories and Morita equivalence to
obtain $\Lang(\Syn(\mathscr{T})) \cong \mathscr{T}$. The fact that the thus
established isomorphisms are natural on $\catC$ and $\mathscr{T}$ follows from
routine calculations.  \end{proof}

\section{From equations to $\mathcal{V}$-equations}
\label{sec:main}
In this section we extend the results of Section~\ref{sec:back} to the setting
of $\V$-equations.

\subsection{Preliminaries} Let $\mathcal{V}$ denote a commutative and unital
quantale, $\otimes :
\mathcal{V} \times \mathcal{V} \to \mathcal{V}$ the corresponding binary
operation, and $k$ the corresponding unit~\cite{paseka00}.  We start by
recalling two definitions concerning ordered
structures~\cite{gierz03,JGL-topology} and then explain their relevance to our
work.

\begin{defi}
	Consider a complete lattice $L$.  For every $x, y \in L$ we say that
	$y$ is \emph{way-below} $x$ (in symbols, $y \ll x$) if for every
	subset $X \subseteq L$ whenever $x \leq \bigvee X$ there exists a
	\emph{finite} subset $A \subseteq X$ such that $y \leq \bigvee A$.
	The lattice $L$ is called \emph{continuous} iff for every $x \in L$,
	\begin{flalign*}
		x = \bigvee \{ y  \mid y \in L\ \text{and} \ y \ll x \}
	\end{flalign*}
\end{defi}

\begin{defi}
	Let $L$ be a complete lattice. A \emph{basis} $B$ of $L$ is a subset
	$B \subseteq L$ such that for every $x \in L$ the set
	$B \cap \{ y \mid y \in L\ \text{and} \ y \ll x \}$ is directed and
	has $x$ as the least upper bound.
\end{defi}

From now on we assume that the underlying lattice of $\mathcal{V}$ is
continuous and has a basis $B$ which is closed under finite joins, the
multiplication of the quantale $\otimes$ and contains the unit $k$. These
assumptions will allow us to work \emph{only} with a specified subset of
$\mathcal{V}$-equations chosen \eg for computational reasons, such as the
\emph{finite} representation of values $q \in \mathcal{V}$. Additionally, note
that even if we take a basis $B$ \emph{not} closed under such operations we can
close it and the resulting set will again be a basis: indeed the directedness
condition follows directly from the closure under finite joins
and~\cite[Proposition I-1.2]{gierz03} whilst the least upper bound condition is
easy to check.

\begin{exa}
	The Boolean quantale $((\{0 \leq 1\}, \vee), \otimes := \wedge)$ is
	\emph{finite} and thus continuous~\cite{gierz03}. Since it is
	continuous, $\{0,1\}$ itself is a basis for the quantale that
	satisfies the conditions above. For the G\"{o}del
	t-norm~\cite{denecke13} $(([0,1], \vee), \otimes := \wedge)$, the
	way-below relation is the strictly-less relation $<$ with the
	exception that $0 < 0$. A basis for the underlying lattice that
	satisfies the conditions above is the set $\mathbb{Q} \cap
	[0,1]$. Note that, unlike real numbers, rationals numbers always
	have a finite representation. For the metric quantale (also known as
	Lawvere quantale) $(([0,\infty], \wedge), \otimes := +)$, the
	way-below relation corresponds to the \emph{strictly greater}
	relation with $\infty > \infty$, and a basis for the underlying
	lattice that satisfies the conditions above is the set of extended
	non-negative rational numbers. The latter also serves as basis for
	the ultrametric quantale $(([0,\infty], \wedge), \otimes := \max)$.
\end{exa}

From now on we assume that $\mathcal{V}$ is \emph{integral}, \ie that the unit
$k$ is the top element of $\mathcal{V}$. Despite not being strictly necessary
in many of our results, it will allow us to establish a smoother theory of
$\mathcal{V}$-equations,
whilst still covering \eg\ all the examples above. This assumption is common in
quantale theory~\cite{velebil19}.

\subsection{A $\V$-equational deductive system}
\label{sec:sys}
As mentioned in the introduction, $\mathcal{V}$ induces the notion of a
$\mathcal{V}$-equation, \ie an equation $t =_q s$ labelled by an element $q$ of
$\mathcal{V}$. This subsection explores this concept by introducing a
$\mathcal{V}$-equational deductive system for linear $\lambda$-calculus and a
notion of a linear $\mathcal{V}\lambda$-theory.  Recall the term formation
rules of linear $\lambda$-calculus from Fig.~\ref{fig:lang}. A
$\mathcal{V}$-equation-in-context is an expression $\Gamma \vljud v =_q w :
\typeA$ with $q \in B$ (the basis of $\mathcal{V}$), $\Gamma \vljud v : \typeA$
and $\Gamma \vljud w : \typeA$. Let $\top$ be the top element in $\mathcal{V}$.
An equation-in-context $\Gamma \vljud v = w : \typeA$ now denotes the
particular case in which both $\Gamma \vljud v =_\top w : \typeA$ and $\Gamma
\vljud w =_\top v : \typeA$. For the case of the Boolean quantale,
$\mathcal{V}$-equations are labelled by $\{0,1\}$. We will see that $\Gamma
\vljud v =_1 w : \typeA$ can be effectively treated as an inequation $\Gamma
\vljud v \leq w : \typeA$, whilst $\Gamma \vljud v =_0 w : \typeA$ corresponds
to a trivial $\mathcal{V}$-equation, \ie a $\mathcal{V}$-equation that always
holds. For the G\"{o}del t-norm, we can choose $\mathbb{Q} \cap [0,1]$ as basis
and then obtain what we call \emph{fuzzy inequations}. For the metric quantale,
we can choose the set of extended non-negative rational numbers as basis and
then obtain \emph{metric equations} in the spirit
of~\cite{mardare16,mardare17}. Similarly, by choosing the ultrametric quantale
$(([0,\infty], \wedge), \otimes := \max)$ with the set of extended non-negative
rational numbers as basis we obtain what we call \emph{ultrametric equations}.

\begin{defi}[$\mathcal{V}\lambda$-theories]\label{defn:vtheory}
	Consider a tuple $(G,\Sigma)$ consisting of a class $G$ of ground
	types and a class of sorted operation symbols
	$f : \typeA_1,\dots,\typeA_n \to \typeA$ with $n \geq 1$. A linear
	$\mathcal{V}\lambda$-theory $((G,\Sigma),Ax)$ is a tuple such that
	$Ax$ is a class of $\mathcal{V}$-equations-in-context over linear
	$\lambda$-terms built from $(G,\Sigma)$.
\end{defi}

\begin{figure}[h!]
\small{
	{\renewcommand{\arraystretch}{1}
	\begin{tabular}{|m{11em} m{11em} m{11em}|}
		\hline
		\multicolumn{3}{|l|}{}
		\\
		\infer[\textbf{(refl)}]{v =_\top v}{}
		&
		\infer[\textbf{(trans)}]{v =_{q \otimes r} u}{
			v =_q w  \qquad
			w =_r u}
		&
                \infer[\textbf{(weak)}]{v =_r w}{v =_q w \qquad r \leq q }
		\\
		\infer[\textbf{(arch)}]{v =_q w}{
			\forall r \ll q .\ v =_r w}
		&
		\infer[\textbf{(join)}]{v =_{\vee q_i} w}{\forall i \leq n.\ v =_{q_i} w}
                & 
                \infer[]{v \otimes v' =_{q \otimes r} w \otimes w'}{
			v =_q w \quad v' =_r w'}
                \\
		\infer[]{f(v_1,\dots,v_n) =_{\otimes q_i} f(w_1,\dots,w_n)}
		{\forall i \leq n.\ v_i =_{q_i} w_i}
		&
		\infer[]{
			v\ \prog{to}\ \ast .\  v'=_{q \otimes r}
			w\ \prog{to}\ \ast .\ w'}
		{v =_q w \qquad v' =_r w'}
		&
                \infer[]{
	        \lambda x : \typeA .\ v =_q  \lambda x :
	        \typeA .\ w}{v =_q w}
		\\
		\multicolumn{2}{|c}{
		\infer[]{
			\prog{pm}\ v\ \prog{to}\ x \otimes y.\ v' =_{q \otimes r}
			\prog{pm}\ w\ \prog{to}\ x \otimes y.\ w'}
			{v =_q w \qquad v' =_r w' }
		}
		&
                \infer[]{v \, v' =_{q \otimes r} w \, w'}
		{v =_q w \quad v' =_r w'}
		\\
                \multicolumn{2}{|c}{ 
                \infer[]{\Delta \vljud v =_q w : \typeA}{
		\Gamma \vljud v =_q w : \typeA \qquad \Delta \in
                \mathrm{perm}(\Gamma)}
                }
                &
		\infer[]{v[v'/x] =_{q \otimes r}w[w'/x]}
		{v =_q w \qquad v' =_r w'}		
                \\
	        \hline
	\end{tabular}
}}	
\caption{$\mathcal{V}$-congruence rules.}
\label{fig:theo_rules}
\end{figure}

The elements of $Ax$ are called axioms of the theory. Let $Th(Ax)$ be the
smallest class that contains $Ax$ and that is closed under the rules of
Fig.~\ref{fig:eqs} (\ie the classical equational system) and of
Fig.~\ref{fig:theo_rules} (as usual we omit the context and typing
information). The elements of $Th(Ax)$ are called theorems of the theory.

Let us examine the rules in Fig.~\ref{fig:theo_rules} in more detail. They can
be seen as a generalisation of the notion of a congruence. The rules
\textbf{(refl)} and \textbf{(trans)} are a generalisation of equality's
reflexivity and transitivity. Rule \textbf{(weak)} encodes the principle that
the higher the label in the $\mathcal{V}$-equation, the `tighter' is the
relation between the two terms in the $\mathcal{V}$-equation. In other words,
$v =_r w$ is subsumed by $v =_q w$, for $r\leq q$. This can be seen clearly \eg
with the metric quantale by reading $v =_q w$ as ``the terms $v$ and $w$ are
\emph{at most} at distance $q$ from each other'' (recall that in the metric
quantale the usual order is reversed, \ie $\leq\ :=\ \geq_{[0,\infty]}$). Rule
\textbf{(arch)} is essentially a generalisation of the Archimedean rule
in~\cite{mardare16,mardare17}. It says that if $v =_r w$ for all
\emph{approximations} $r$ of $q$ then it is also the case that $v =_q w$. Rule
\textbf{(arch)} involves possibly infinitely many premisses. We comment on its
necessity after establishing the aforementioned soundness and completeness
theorem (\cref{theo:sound_compl2}).  Rule \textbf{(join)} says that deductions
are closed under finite joins, and in particular it is always the case that $v
=_\bot w$. All other rules correspond to a generalisation of
\emph{compatibility} to a $\mathcal{V}$-equational setting.  The reader may
have noticed that the rules in Fig.~\ref{fig:theo_rules} do not contain a
$\mathcal{V}$-generalisation of symmetry w.r.t.  standard equality. Such a
generalisation would be:
\begin{flalign*}
	\infer{w =_q v}{v =_q w}
\end{flalign*}
This rule is not present in Fig.~\ref{fig:theo_rules} because in some
quantales $\mathcal{V}$ it forces too many
$\mathcal{V}$-equations. For example, in the Boolean quantale the condition
$v \leq w$ would automatically entail $w \leq v$ (due to
symmetry); in fact, for this particular case symmetry forces the notion
of inequation to collapse into the classical notion of equation. On
the other hand, symmetry is desirable in the (ultra)metric case
because (ultra)metrics need to respect the symmetry
equation~\cite{JGL-topology}.
\begin{defi}[Symmetric linear $\mathcal{V}\lambda$-theories]
	A symmetric linear $\mathcal{V}\lambda$-theory is a linear
	$\mathcal{V}\lambda$-theory whose set of theorems is closed under
	symmetry.
\end{defi}

We end this subsection by briefly commenting on linear
$\mathcal{V}\lambda$-theories where $\mathcal{V}$ is a quantale with a
\emph{linear order}. This is relevant for comparing our general $\V$-equational
system to that of metric equations~\cite{mardare16,mardare17} which tacitly
uses a quantale with a linear order, namely  the metric quantale. 

\begin{thm}\label{theo:simpl}
  Assume that the underlying order of $\mathcal{V}$ is linear and
  consider a (symmetric) linear
  $\mathcal{V}\lambda$-theory. Substituting the rule below on the left
  by the one below on the right does not change the theory.
  \begin{flalign*}
    \infer{v =_{\vee q_i} w}{\forall i \leq n.\ v =_{q_i} w} \hspace{2cm}
    \infer{v =_\bot w}{}
  \end{flalign*}
\end{thm}

\begin{proof}
  Clearly, the rule on the left subsumes the one on the right by
  choosing $n = 0$. So we only need to show the inverse direction
  under the assumption that $\mathcal{V}$ is linear. Thus, assume that
  $\forall i \leq n.\ v =_{q_i} w$. We proceed by case distinction. If
  $n = 0$ then we need to show that $v =_\bot w$ which is given
  already by the rule on the right.  Suppose now that $n > 0$. Then
  since the order of $\mathcal{V}$ is linear the value $\vee q_i$ must
  already be one of the values $q_i$ and $v =_{q_i} w$ is already part
  of the theory. In other words, in case of $n > 0$ the rule on the
  left is redundant.
\end{proof}

The above result is in accordance with metric universal
algebra~\cite{mardare16,mardare17} which also does not include rule
\textbf{(join)}. Interestingly, however, we still have $v =_\bot w$ for all
$\lambda$-terms $v$ and $w$ and the counterpart of such a rule is not present
in~\cite{mardare16,mardare17}. This is explained by the fact that metric
equations in~\cite{mardare16,mardare17} are labelled \emph{only} by
non-negative rational numbers whilst we also permit infinity to be a label (in
our case, labels are given by a basis $B$ which for the metric case corresponds
to the \emph{extended} non-negative rational numbers). All remaining rules of
our $\mathcal{V}$-equational system instantiated to the metric case find a
counterpart in the metric equational system presented
in~\cite{mardare16,mardare17}.  Similarly, when $\mathcal{V}$ is the Boolean
quantale (which is linear) and we assume $\mathcal{V}\lambda$-theories to be
symmetric $\mathcal{V}$-congruence is equivalent to classical congruence.

\subsection{Semantics of $\mathcal{V}$-equations}
\label{sec:vcat}
In this subsection we set the necessary background for presenting a
sound and complete class of models for (symmetric) linear
$\mathcal{V}\lambda$-theories. We start by recalling basics concepts
of $\mathcal{V}$-categories, which are central in a field initiated by
Lawvere in~\cite{lawvere73} and can be intuitively seen as generalised
metric spaces~\cite{stubbe14,hofmann20,velebil19}. As we will show,
$\mathcal{V}$-categories provide structure to suitably interpret
$\mathcal{V}$-equations.

\begin{defi}
  \label{defn:vcat}
  A (small) $\mathcal{V}$-category is a pair $(X,a)$ where $X$ is a
  class (set) and $a : X \times X \to \mathcal{V}$ is a function that
  satisfies:
  \begin{flalign*}
    k \leq a(x,x) \qquad \text{ and }  \qquad
    a(x,y) \otimes a(y,z) \leq a(x,z) \hspace{2cm}
    (x,y,z \in X)
  \end{flalign*}
  For two $\mathcal{V}$-categories $(X,a)$ and $(Y,b)$, a
  $\mathcal{V}$-functor $f : (X,a) \to (Y,b)$ is a function
  $f : X \to Y$ that satisfies the inequality
  $a(x,y) \leq b(f(x),f(y))$ for all $x,y \in X$.
\end{defi}
Small $\mathcal{V}$-categories and $\mathcal{V}$-functors form a
category which we denote by $\VCat$.  A $\mathcal{V}$-category $(X,a)$
is called \emph{symmetric} if $a(x,y) = a(y,x)$ for all $x,y \in
X$. We denote by $\VCatSy$ the full subcategory of $\VCat$ whose
objects are symmetric. Every $\mathcal{V}$-category carries a natural
order defined by $x \leq y$ whenever $k \leq a(x,y)$. A
$\mathcal{V}$-category is called \emph{separated} if its natural order
is anti-symmetric. We denote by $\VCatSe$ the full subcategory of
$\VCat$ whose objects are separated.
\begin{exa}
  For $\mathcal{V}$ the Boolean quantale, $\VCatSe$ is the category $\Pos$ of
  partially ordered sets and monotone maps; $\VCatSS$ is simply the category
  $\Set$ of sets and functions.  For $\mathcal{V}$ the metric quantale,
  $\VCatSS$ is the category $\Met$ of extended metric spaces and non-expansive
  maps. For $\mathcal{V}$ the ultrametric quantale, $\VCatSS$ is the category
  of extended ultrametric spaces and non-expansive maps. In what follows we
  omit the qualifier `extended' in `extended (ultra)metric spaces'. 
\end{exa}
The inclusion functor $\VCatSe \hookrightarrow \VCat$ has a left
adjoint \cite{hofmann20}. It is constructed first by defining the
equivalence relation $x \sim y$ whenever $x \leq y$ and $y \leq x$
(for $\leq$ the natural order introduced earlier). Then this relation
induces the separated $\mathcal{V}$-category $(X/_\sim, \tilde a)$
where $\tilde a$ is defined as $\tilde a([x],[y]) = a(x,y)$ for every
$[x],[y] \in X/_\sim$. The left adjoint of the inclusion functor
$\VCatSe \hookrightarrow \VCat$ sends every $\mathcal{V}$-category
$(X,a)$ to $(X/_\sim, \tilde a)$. This quotienting construct preserves
symmetry, and therefore we automatically obtain the following result.
\begin{thm}\label{theo:reflective}
  The inclusion functor $\VCatSS \hookrightarrow \VCatSy$ has a left
  adjoint.
\end{thm}
Next, we recall notions of enriched category theory~\cite{kelly82} instantiated
into the setting of \emph{autonomous categories enriched over
$\mathcal{V}$-categories}.  We will use the enriched structure to give
semantics to $\mathcal{V}$-equations between linear $\lambda$-terms (the latter
as usual being morphisms in the category itself).  First, note that every
category $\VCat$ is autonomous with the tensor $(X,a) \otimes (Y,b) := (X
\times Y, a \otimes b)$ where $a\otimes b$ is defined as,
\begin{flalign*}
  (a \otimes b)((x,y), (x',y')) = a(x,x') \otimes b(y,y')
\end{flalign*}
and the set of $\mathcal{V}$-functors $\VCat((X,a),(Y,b))$ is equipped
with the map,
\[
(f,g) \mapsto \bigwedge_{x \in X} b(f(x),g(x))
\]
\begin{thm}\label{theo:auto}
  The categories $\VCatSy$,
  $\VCatSe$, and $\VCatSS$ inherit the autonomous structure of
  $\VCat$ whenever $\mathcal{V}$ is integral.
\end{thm}

\begin{proof}
  The proof follows by showing that the closed monoidal structure of
  $\VCat$ preserves symmetry and separation. It is immediate for
  symmetry. For separation, note that since $\mathcal{V}$ is integral
  the inequation $x \otimes y \leq x$ holds for all
  $x,y \in \mathcal{V}$. It follows that the monoidal structure
  preserves separation. The fact that the closed structure also
  preserves separation uses the implication
  $x \leq \bigwedge A \Rightarrow \forall a \in A.\ x \leq a$ for all
  $x \in X, A \subseteq X$.
\end{proof}

Since we assume that $\mathcal{V}$ is integral, this last theorem states that
all the categories mentioned therein are suitable bases of
enrichment~\cite{kelly82}. 

\begin{defi}\label{def:VCatEnriched}
  A category $\catC$ is $\VCat$-enriched (or simply, a $\VCat$-category) if for
  all $\catC$-objects $X$ and $Y$ the hom-set $\catC(X,Y)$ is a
  \emph{$\mathcal{V}$-category} and if the composition of $\catC$-morphisms, 
  \begin{flalign*}
    (\ \cdot\ ) : \catC(X,Y) \otimes \catC(Y,Z)
    \longrightarrow \catC(X,Z)
  \end{flalign*}
  is a $\mathcal{V}$-functor. Given two $\VCat$-categories $\catC$ and $\catD$
  and a functor $\funF : \catC \to \catD$, we call $\funF$ a $\VCat$-enriched
  functor (or simply, $\VCat$-functor) if for all $\catC$-objects $X$ and $Y$
  the map $\funF_{X,Y} : \catC(X,Y) \to \catD(\funF X, \funF, Y)$ is a
  $\mathcal{V}$-functor.  An adjunction $\catC : \funF \dashv \funG : \catD$ is
  called $\VCat$-enriched if for all objects $X \in |\catC|$ and $Y \in
  |\catD|$ there exists a $\mathcal{V}$-isomorphism $\catD(\funF X, Y) \cong
  \catC(X, \funG Y)$ natural in $X$ and $Y$.  We obtain analogous notions of
  enrichment by substituting $\VCat$ with $\VCatSe$, $\VCatSy$, or $\VCatSS$.
\end{defi}

If $\catC$ is a $\VCat$-category then $\catC \times \catC$ is also a
$\VCat$-category via the tensor operation $\otimes$ in $\VCat$.  We take
advantage of this fact in the following definition.

\begin{defi}\label{defn:enr_aut}
  A $\VCat$-enriched autonomous category $\catC$ is an autonomous and
  $\VCat$-enriched category $\catC$ such that the bifunctor $\otimes : \catC
  \times \catC \to \catC$ is a $\VCat$-functor and the adjunction $(- \otimes
  X) \dashv (X \multimap -)$ is a $\VCat$-adjunction.  We obtain analogous
  notions of enriched autonomous category by replacing $\VCat$ (as basis of
  enrichment) with $\VCatSe$, $\VCatSy$, or $\VCatSS$.
\end{defi}

\begin{exa}\label{ex:pos_met}
  Recall that $\Pos \cong \VCatSe$ when $\mathcal{V}$ is the Boolean quantale.
  According to~\cref{theo:auto} the category $\Pos$ is autonomous. It
  follows by general results that the category is
  $\Pos$-enriched~\cite{borceux94}. It is also easy to see that its tensor is
  $\Pos$-enriched and that the adjunction $(- \otimes X) \dashv (X \multimap
  -)$ is $\Pos$-enriched. Therefore, $\Pos$ is an instance of
  Definition~\ref{defn:enr_aut}. Note also that $\Set \cong \VCatSS$ for
  $\mathcal{V}$ the Boolean quantale and that $\Set$ is trivially an instance
  of Definition~\ref{defn:enr_aut}.

  Recall that $\Met \cong \VCatSS$ when $\mathcal{V}$ is the metric quantale.
  Thus, the category $\Met$ is autonomous (\cref{theo:auto}) and
  $\Met$-enriched~\cite{borceux94}.  It follows as well from routine
  calculations that its tensor is $\Met$-enriched and that the adjunction $(-
  \otimes X) \dashv (X \multimap -)$ is $\Met$-enriched. Therefore $\Met$ is an
  instance of Definition~\ref{defn:enr_aut}. An analogous reasoning tells that
  the category of ultrametric spaces (enriched over itself) is also an instance
  of Definition~\ref{defn:enr_aut}. We present further examples of
  $\VCat$-enriched autonomous categories in Section~\ref{sec:examples}.
\end{exa}

Finally, let us recall the interpretation of linear $\lambda$-terms on an
autonomous category $\catC$ (see Section~\ref{sec:back}) and assume that
$\catC$ is $\VCat$-enriched. Then we say that a $\mathcal{V}$-equation $\Gamma
\vljud v =_q w : \typeA$ is \emph{satisfied} by this interpretation if $q \leq
a(\sem{\Gamma \vljud v : \typeA},\sem{\Gamma \vljud w : \typeA})$ where $a :
\catC(\sem{\Gamma},\sem{\typeA}) \times \catC(\sem{\Gamma},\sem{\typeA}) \to
\mathcal{V}$ is the underlying function of the $\mathcal{V}$-category
$\catC(\sem{\Gamma},\sem{\typeA})$.

\begin{thm}\label{theo:sound}
  The rules listed in Fig.~\ref{fig:eqs} and in Fig.~\ref{fig:theo_rules}
  are sound for $\VCat$-enriched autonomous categories $\catC$.
  Specifically if $\Gamma \vljud v =_q w : \typeA$ results from the
  rules in Fig.~\ref{fig:eqs} and Fig.~\ref{fig:theo_rules} then
  $q \leq a(\sem{\Gamma \vljud v : \typeA},\sem{\Gamma \vljud w : \typeA})$.
\end{thm}

\begin{proof}
  Let us focus first on the equations listed in Fig.~\ref{fig:eqs}.  Recall
  that an equation $\Gamma \vljud v = w : \typeA$ abbreviates the
  $\mathcal{V}$-equations $\Gamma \vljud v =_\top w : \typeA$ and $\Gamma
  \vljud w =_\top v : \typeA$.  Moreover, we already know that the equations
  listed in Fig.~\ref{fig:eqs} are sound for autonomous categories,
  specifically if $v = w$ is an equation of Fig.~\ref{fig:eqs} then $\sem{v} =
  \sem{w}$ in $\catC$ (\cref{theo:sound_compl}).  Thus, by the
  definition of a $\mathcal{V}$-category (Definition~\ref{defn:vcat}) and by
  the assumption of $\mathcal{V}$ being integral ($k = \top$) we obtain
  $\top = k \leq a(\sem{v},\sem{w})$ and $\top = k \leq a(\sem{w},\sem{v})$.

  Let us now focus on the rules listed in Fig.~\ref{fig:theo_rules}. The first
  three rules follow from the definition of a $\mathcal{V}$-category and the
  transitivity property of $\leq$. Rule \textbf{(arch)} follows from the
  continuity of $\mathcal{V}$, specifically from the fact that $q$ is the
  \emph{least} upper bound of all elements $r$ that are way-below $q$.  Rule
  \textbf{(join)} follows from the definition of least upper bound.  The
  remaining rules follow from the definition of the tensor functor $\otimes$ in
  $\VCat$, the fact that $\catC$ is $\VCat$-enriched, $\otimes : \catC \times
  \catC \to \catC$ is a $\VCat$-functor, and the fact that $(- \otimes X)
  \dashv (X \multimap -)$ is a $\VCat$-adjunction. For example, for the
  compatibility rule concerning $\rulename{ax}$ we reason as follows:
  \begin{flalign*}
    & \, a(\sem{f(v_1,\dots,v_n)}, \sem{f(w_1,\dots,w_n)}) \\
    & = a(\sem{f} \comp (\sem{v_1} \otimes \dots \otimes \sem{v_n})
    \comp \spl_{\Gamma_1;\dots;\Gamma_n} \comp \sh_E,
    \sem{f} \comp (\sem{w_1} \otimes \dots \otimes \sem{w_n}) \comp
    \spl_{\Gamma_1;\dots;\Gamma_n} \comp \sh_E) & \\
    & \geq a(\sem{f} \comp (\sem{v_1} \otimes \dots \otimes \sem{v_n}),
    \sem{f} \comp (\sem{w_1} \otimes \dots \otimes \sem{w_n})) & \\
    & \geq a(\sem{v_1} \otimes \dots \otimes \sem{v_n}),
      (\sem{w_1} \otimes \dots \otimes \sem{w_n}) & \\
    & \geq a(\sem{v_1},\sem{w_1}) \otimes \dots \otimes
      a(\sem{v_n},\sem{w_n}) &  \\
    & \geq q_1 \otimes \dots \otimes q_n &  
  \end{flalign*}
  where the second step follows from the fact that
  $\spl_{\Gamma_1;\dots;\Gamma_n} \comp \sh_E$ is a morphism in
  $\catC$ and that $\catC$ is $\VCat$-enriched; the third step follows
  from an analogous reasoning; the fourth step follows from the fact that
  $\otimes : \catC \times \catC \to \catC$ is a $\VCat$-functor; the
  last step follows from the premise of the rule in question.  As
  another example, the proof for the substitution rule  proceeds similarly:
  \begin{flalign*}
    & \, a(\sem{v[v'/x]}, \sem{w[w'/x]}) & \\
    & = a(\sem{v} \comp \join_{\Gamma,\typeA} \comp
      (\id \otimes \sem{v'}) \comp \spl_{\Gamma;\Delta},
      \sem{w} \comp \join_{\Gamma,\typeA} \comp
          (\id \otimes \sem{w'}) \comp \spl_{\Gamma;\Delta}) & \\
    & \geq a(\sem{v} \comp \join_{\Gamma,\typeA} \comp
      (\id \otimes \sem{v'}),
      \sem{w} \comp \join_{\Gamma,\typeA} \comp
          (\id \otimes \sem{w'})) & \\
    & \geq
    a(\id \otimes \sem{v'}, \id \otimes \sem{w'}) \otimes
    a(\sem{v} \comp \join_{\Gamma,\typeA},
          \sem{w} \comp \join_{\Gamma,\typeA})&  \\
    & \geq
    a(\id \otimes \sem{v'}, \id \otimes \sem{w'}) \otimes
      a(\sem{v}, \sem{w}) \\
    & \geq
    a(\id,\id) \otimes a(\sem{v'}, \sem{w'})  \otimes
    a(\sem{v}, \sem{w}) \\
    & =
          a(\sem{v'}, \sem{w'}) \otimes a(\sem{v}, \sem{w})&  \\
    & \geq q \otimes r
  \end{flalign*}
  The proof for the rule concerning $(\multimap_i)$ additionally
  requires the following two facts: if a $\mathcal{V}$-functor
  $f : (X,a) \to (Y,b)$ is an isomorphism then
  $a(x,x') = b(f(x),f(x'))$ for all $x,x' \in X$. For a context
  $\Gamma$, the morphism
  $\join_{\Gamma; x : \typeA} : \sem{\Gamma} \otimes \sem{\typeA} \to
  \sem{\Gamma, x : \typeA}$ is an isomorphism in $\catC$. The proof
  for the rule concerning the permutation of variables (exchange) also
  makes use of the fact that $\sem{\Delta} \to \sem{\Gamma}$ is
  an isomorphism.
\end{proof}

\subsection{$\V\lambda$-calculus, soundness, and completeness}

We now establish a soundness and completeness result for $\V\lambda$-calculus
introduced in Section~\ref{sec:sys}.  A key construct in this result is the
quotienting of a $\mathcal{V}$-category into a \emph{separated}
$\mathcal{V}$-category: we will use it to identify linear $\lambda$-terms when
generating a syntactic category (from a linear $\mathcal{V}\lambda$-theory)
that satisfies the axioms of autonomous categories. This naturally leads to the
following notion of a model for linear $\mathcal{V}\lambda$-theories.

\begin{defi}[Models of linear $\mathcal{V}\lambda$-theories]
  Consider a linear $\mathcal{V}\lambda$-theory $((G,\Sigma),Ax)$ and
  a $\VCatSe$-enriched autonomous category $\catC$. Suppose that for
  each $X \in G$ we have an interpretation $\sem{X}$ as a
  $\catC$-object and analogously for the operation symbols. This
  interpretation structure is a model of the theory if all axioms in
  $Ax$ are satisfied by the interpretation.
\end{defi}

Let us then focus on establishing a completeness result for
$\V\lambda$-calculus. For two types $\typeA$ and $\typeB$ of a
$\mathcal{V}\lambda$-theory $\mathscr{T}$, consider the class
$\closValBP{\typeA}{\typeB}$ of values $v$ such that $x : \typeA \vljud v :
\typeB$. We equip $\closValBP{\typeA}{\typeB}$ with the function $a :
\closValBP{\typeA}{\typeB} \times \closValBP{\typeA}{\typeB} \to \mathcal{V}$
defined by,
\[
        a(v,w) = \bigvee \{ q \mid v =_q w \text{ is a theorem of } \mathscr{T}\}
\]
It is easy to see that $(\closValBP{\typeA}{\typeB},a)$ is a (possibly large)
$\mathcal{V}$-category. We then quotient this $\mathcal{V}$-category into a
\emph{separated} $\mathcal{V}$-category which we denote by
$(\closValBP{\typeA}{\typeB},a)_{/ \sim}$. As detailed in the proof of the next
theorem, this $\V$-category will serve as a hom-object of the syntactic
category $\Syn(\mathscr{T})$ generated from the linear
$\mathcal{V}\lambda$-theory $\mathscr{T}$.  We call $\mathscr{T}$
\emph{varietal} if $(\closValBP{\typeA}{\typeB},a)_{/ \sim}$ is a \emph{small}
$\mathcal{V}$-category for all types $\typeA$ and $\typeB$. In the rest of the
paper we will only work with varietal theories and locally small categories.

\begin{thm}[Soundness \& Completeness]\label{theo:sound_compl2}
  Consider a varietal $\mathcal{V}\lambda$-theory. A
  $\mathcal{V}$-equation-in-context $\Gamma \vljud v =_q w : \typeA$
  is a theorem iff it holds in all models of the theory.
\end{thm}

\begin{proof}[Proof]
  Soundness follows by induction over the rules that define the class
  $Th(Ax)$ and by an appeal to~\cref{theo:sound}.  For completeness, we use a
  strategy similar to the proof of~\cref{theo:sound_compl}, and take
  advantage of the quotienting of a $\mathcal{V}$-category into a separated
  $\mathcal{V}$-category. Recall that we assume that the theory is
  \emph{varietal} and therefore can safely take
  $(\closValBP{\typeA}{\typeB},a)_{/ \sim}$  (defined above) to be a small
  $\mathcal{V}$-category. Note that the quotienting process identifies all
  terms $x : \typeA \vljud v : \typeB$ and $x : \typeA \vljud w : \typeB$
  such that $v =_\top w$ and $w =_\top v$. Such a relation contains the
  equations-in-context from Fig.~\ref{fig:eqs} and moreover it is
  straightforward to show that it is compatible with the term formation rules
  of linear $\lambda$-calculus (Fig.~\ref{fig:lang}). So, analogously
  to~\cref{theo:sound_compl} we obtain an autonomous category
  $\Syn(\mathscr{T})$ whose objects are the types of the language and whose
  hom-sets are the underlying sets of the $\mathcal{V}$-categories
  $(\closValBP{\typeA}{\typeB},a)_{/ \sim}$.

  Our next step is to show that the category $\Syn(\mathscr{T})$ has a
  $\VCatSe$-enriched autonomous structure.  We start by showing that the
  composition map
  $\Syn(\mathscr{T})(\typeA,\typeB) \otimes \Syn(\mathscr{T})(\typeB,\typeC) \to
  \Syn(\mathscr{T})(\typeA,\typeC)$ is a $\mathcal{V}$-functor:
  \begin{flalign*}
    & \,  a(([v'],[v]), ([w'],[w])) & \\
    & = a([v],[w]) \otimes a([v'],[w']) &  \\
    & = a(v,w) \otimes a(v',w') & \\
    & = \bigvee \{ q \mid v =_q w \} \otimes
      \bigvee \{ r \mid v' =_r w' \}
    &  \\
    & = \bigvee \{ q \otimes r \mid v =_q w, v' =_r w' \} & 
    \{\text{Defn. of quantale}\}\\
    & \leq \bigvee \{ q \mid v[v'/x]
      =_q w[w'/x] \}  &
      \left \{ A \subseteq B \Rightarrow \bigvee A \leq \bigvee B 
      \right \} \\
                      & = a( v[v'/x] , w[w'/x]) & \\
                      & = a( [v[v'/x]] , [w[w'/x]]) & \\
                      & = a([v] \comp [v'] ,[w] \comp [w']) &
  \end{flalign*}
  The fact that $\otimes : \Syn(\mathscr{T}) \times \Syn(\mathscr{T}) \to
  \Syn(\mathscr{T})$ is a $\VCat$-functor follows by an analogous reasoning.
  Next, we need to show that $(- \otimes X) \dashv (X \multimap -)$ is a
  $\VCat$-adjunction. It is straightforward to show that both functors are
  $\VCat$-functors, and from a similar reasoning it follows that the
  isomorphism $\Syn(\mathscr{T})(\typeB, \typeA \multimap \typeC) \cong 
  \Syn(\mathscr{T})(\typeB
  \otimes \typeA, \typeC)$ is a $\mathcal{V}$-isomorphism.

  The final step is to show that if an equation
  $\Gamma \vljud v =_q w : \typeA$ with $q \in B$ is satisfied by
  $\Syn(\mathscr{T})$ then it is a theorem of the linear
  $\mathcal{V}\lambda$-theory. By assumption
  $a([v],[w]) = a(v,w) = \bigvee \{ r \mid v =_r w \} \geq q$. It follows
  from the definition of the way-below relation that for all $x \in B$
  with $x \ll q$ there exists a \emph{finite} set
  $A \subseteq \{ r \mid v =_r w \}$ such that $x \leq \bigvee A$. Then
  by an application of rule \textbf{(join)} in Fig.~\ref{fig:theo_rules}
  we obtain $v =_{\bigvee A} w$, and consequently rule \textbf{(weak)} in
  Fig.~\ref{fig:theo_rules} provides $v =_x w$ for all $x \ll
  q$. Finally, by an application of rule \textbf{(arch)} in
  Fig.~\ref{fig:theo_rules}  we deduce that
  $v =_q w$ is part of the theory.
\end{proof}

\cref{theo:sound_compl2} extends straightforwardly to symmetric linear $\V
\lambda$-theories and $\VCatSS$-autonomous categories. In particular, when
$\mathcal{V}$ is the Boolean quantale we obtain the classical soundness and
completeness result described in the previous section. Additionally, note that
by inspecting the last part of the proof of \cref{theo:sound_compl2} it is
clear that Rule~\textbf{(arch)} takes a crucial role. It is possible however to
drop the latter and prove the following variation of completeness.

\begin{thm}[Approximate completeness]
        Consider a varietal $\V \lambda$-theory $\mathscr{T}$. If a
        $\V$-equation $\Gamma \vljud v =_q w : \typeA$ holds in all models of
        the theory then for all approximations $r \ll q$ with $r \in B$ we have
        $\Gamma \vljud v =_r w : \typeA$ as a theorem. In particular if $q$ is
        compact (\ie $q \ll q$) we have $\Gamma \vljud v =_q w : \typeA$ as a
        theorem.
\end{thm}
\subsection{Equivalence theorem between linear $\mathcal{V}\lambda$-theories and
$\VCatSe$-autonomous categories}
\label{sec:equiv}
Analogously to Section~\ref{sec:eq} (where we address $\lambda$-calculus with
just classical equations), in this subsection we will present a category of
$\VCatSe$-autonomous categories, a category of linear $\V \lambda$-theories,
and then prove the existence of an equivalence between both categories.  For
reasons analogous to those presented in Section~\ref{sec:eq}, we start by
providing a bijective correspondence up-to isomorphism between models of $\V
\lambda$-theories $\mathscr{T}$ and $\VCatSe$-autonomous functors
$\Syn(\mathscr{T}) \to \catC$. The following proposition will help us establish
that.

\begin{prop}
        \label{prop:modelV}
        Consider a $\VCat$-autonomous functor $\funF : \catC \to \catD$ and a model
        $M$ on $\catC$ of a $\V \lambda$-theory $\mathscr{T}$. Then
        $\funF_\ast M$ is a model of $\mathscr{T}$ on $\catD$.
\end{prop}

\begin{proof}
        Recall from the proof of Proposition~\ref{prop:model} the equation,
        \[
                \sem{\Gamma \vljud v : \typeA}_{\funF_\ast M} = 
                h^{-1}_\typeA \comp 
                \funF \sem{\Gamma \vljud v : \typeA}_{M} \comp \mu[\Gamma] 
                \comp h[\Gamma]
        \]
        which holds for all judgements $\Gamma \vljud v : \typeA$. We use it to
        reason in the following manner:
\begin{flalign*}
        & \, v =_q w  &  \\
        & \Rightarrow a(\sem{v}_M, \sem{w}_M) \geq q 
        & \text{\{$M$ is a model of } \mathscr{T} \} \\
        & \Rightarrow a(\funF \sem{v}_M, \funF \sem{w}_M) \geq q
        & \text{\{$\funF$ is $\VCat$-enriched\}} \\
        & \Rightarrow a(h^{-1} \comp \funF \sem{v}_M, h^{-1} \comp 
                \funF \sem{w}_M) \geq q
        & \text{\{$\catC$ is $\VCat$-enriched\}} \\
        & \Rightarrow a(h^{-1} \comp \funF \sem{v}_M \comp \mu[\Gamma] 
                \comp h[\Gamma], h^{-1} \comp 
                \funF \sem{w}_M \comp \mu[\Gamma] \comp h[\Gamma] ) \geq q
        & \text{\{$\catC$ is $\VCat$-enriched\}} \\
        & \Rightarrow a(\sem{v}_{\funF_\ast M}, \sem{w}_{\funF_\ast M}) 
        \geq q &  \qedhere
\end{flalign*}

\end{proof}

\begin{thm} \label{th:strV}
        Let $\mathscr{T}$ be a $\V\lambda$-theory. Every $\VCatSe$-autonomous
        functor $\funF : \Syn(\mathscr{T}) \to \catC$ induces a model
        ${\funF_\ast \Syn(\mathscr{T})}$ of $\mathscr{T}$ and every model $M$
        of $\mathscr{T}$ induces a $\VCatSe$-strict autonomous functor $([v]
        \mapsto \sem{v}_M) : \Syn(\mathscr{T}) \to \catC$. Furthermore these
        constructions are inverse to each other up-to isomorphism, in the sense
        that, 
        \begin{align*}
                ([v] \mapsto \sem{v}_{\funF_\ast \Syn(\mathscr{T})}) 
                \cong \funF \hspace{1cm} 
                \text{ and } \hspace{1cm}
                ([v] \mapsto \sem{v}_M)_\ast \Syn(\mathscr{T}) = M
        \end{align*}
\end{thm}

\begin{proof}
       Let  us first focus on the mapping that sends functors to models.
       Consider a $\VCat$-autonomous functor $\funF : \Syn(\mathscr{T}) \to
       \catC$.  Then observe that $\Syn(\mathscr{T})$ corresponds to a model of
       $\mathscr{T}$ and thus by Proposition~\ref{prop:modelV} we conclude that
       $\funF_\ast \Syn(\mathscr{T})$ must be a model of $\mathscr{T}$. For the
       inverse direction,  consider a model of $\mathscr{T}$ over $\catC$. Let
       $a$ denote the underlying function of the hom-($\mathcal{V}$-categories)
       in $\Syn(\mathscr{T})$ and $b$ the underlying function of the
       hom-($\mathcal{V}$-categories) in $\catC$. Then note that if $[v] = [w]$
       in $\Syn(\mathscr{T})$ then, by completeness, the equations $v =_\top w$
       and $w =_\top v$ are theorems of $\mathscr{T}$, which means that
       $\sem{v}_M = \sem{w}_M$ by the definition of a model and
       \emph{separability}. This allows us to define a mapping $F :
       \Syn(\mathscr{T}) \to \catC$ that sends each type $\typeA$ to
       $\sem{\typeA}_M$ and each morphism $[v]$ to $\sem{v}_M$. The fact that
       this mapping is an autonomous functor follows from an analogous
       reasoning to the one used in the proof of~\cref{th:str}.  We now need to
       show that this functor is $\VCatSe$-enriched. Recall that $a([v],[w]) =
       \bigvee \{ q \mid v =_q w \}$ and observe that for every $v =_q w$ in
       the previous quantification we have $b(\sem{v}_M,\sem{w}_M) \geq q$ (by
       the definition of a model), which establishes, by the definition of a
       least upper bound, $a([v],[w]) = \bigvee \{ q \mid v =_q w \} \leq
       b(\sem{v}_M,\sem{w}_M)$.  Finally the proof for showing that,
       \begin{align*}
                ([v] \mapsto \sem{v}_{\funF_\ast \Syn(\mathscr{T})}) 
                \cong \funF \hspace{1cm} 
                \text{ and } \hspace{1cm}
                ([v] \mapsto \sem{v}_M)_\ast \Syn(\mathscr{T}) = M
        \end{align*}
        is completely analogous to that of~\cref{th:str}.
\end{proof}

Next, let $\VCatSe$-$\Aut$ be the category of $\VCatSe$-autonomous
categories and $\VCatSe$-autonomous functors. Consider then
$\VCatSe$-$\Aut_{/_{\cong}}$ whose objects are $\VCatSe$-autonomous categories
and morphisms are \emph{isomorphism classes} of $\VCatSe$-autonomous functors.
We now show that this category is equivalent to
the category $\V \lambda\text{-} \catfont{Th}$ whose objects are
linear $\V \lambda$-theories,
\begin{align} \label{equ}
        \xymatrix@C=40pt{
                \V \lambda\text{-} \catfont{Th}
                \ar@<5pt>[r]_(0.40){\simeq} & 
                \VCatSe\text{-}\Aut_{/_{\cong}}
                \ar@<5pt>[l]^(0.60){}
        }
\end{align}
Analogously to Section~\ref{sec:eq}, we set
$\V\lambda\text{-}\catfont{Th}(\mathscr{T}_1,\mathscr{T}_2) :=
\VCatSe\text{-}\Aut_{/_{\cong}}(\Syn(\mathscr{T}_1),\Syn(\mathscr{T}_2))$.  In
words, this means that a morphism $\mathscr{T}_1 \to \mathscr{T}_2$ between
$\V\lambda$-theories $\mathscr{T}_1$ and $\mathscr{T}_2$ is exactly an
isomorphism class of $\VCatSe$-autonomous functors $\Syn(\mathscr{T}_1) \to
\Syn(\mathscr{T}_2)$, which by~\cref{th:strV} are in bijective correspondence
(up-to isomorphism) to models of $\mathscr{T}_1$ on the category
$\Syn(\mathscr{T}_2)$. Note that our remarks in Section~\ref{sec:eq} about
Morita equivalence hold here as well, with the exception that now an
equivalence of theories $\mathscr{T}_1$ and $\mathscr{T}_2$ translates into a
$\VCatSe$-equivalence of categories $\Syn(\mathscr{T}_1)$ and
$\Syn(\mathscr{T}_2)$ instead of just an ordinary equivalence.

The next step is to set down the necessary constructions for describing the
functor $(\VCatSe)$-$\Aut_{/_{\cong}}\to \V\lambda\text{-} \catfont{Th}$ that
is part of the equivalence~\eqref{equ}. 

\begin{defi}[Internal language]
Consider a $\VCatSe$-enriched autonomous category $\catC$. It induces a linear
$\mathcal{V}\lambda$-theory $\Lang(\catC)$ whose ground types and operations
symbols are defined as in the case of linear $\lambda$-theories (recall
Definition~\ref{def:lng}). The axioms of $\Lang(\catC)$ are all the
$\mathcal{V}$-equations-in-context that are satisfied by the obvious
interpretation on $\catC$ plus the $\V$-equations that mark $i_\typeA : \typeA
\to i(\typeA)$ as an isomorphism in the theory $\Lang(\catC)$. 
\end{defi}
In conjunction with the proof of~\cref{theo:sound_compl2}, a consequence
of the following theorem is that $\Syn(\Lang(\catC))$ is a $\VCatSe$-enriched
category. We will need this to properly formulate the claim that $\Lang(\catC)$
completely describes $\catC$ from a syntactical perspective, like we did
in Section~\ref{sec:eq} for classical equations.

\begin{thm}\label{theo:vari}
  The linear $\mathcal{V}\lambda$-theory $\Lang(\catC)$ is varietal.
\end{thm}

\begin{proof}
Let us denote by $\Lang^\lambda(\catC)$ the \emph{linear $\lambda$-theory}
generated from $\catC$. Then according to~\cref{th:int}, the category
$\Syn(\Lang^\lambda(\catC))$ (\ie the syntactic category generated from
$\Lang^\lambda(\catC)$) is locally small whenever $\catC$ is locally small.
Consider now two types $\typeA$ and $\typeB$.  We will prove our claim by
taking advantage of the axiom of replacement in ZF set-theory, specifically by
presenting a \emph{surjective} map, \begin{flalign*}
\Syn(\Lang^\lambda(\catC))(\typeA,\typeB) \twoheadrightarrow
\Syn(\Lang(\catC))(\typeA,\typeB) \end{flalign*} The crucial observation is
that if $v = w$ in $\Lang^\lambda(\catC)$ then $v =_\top w$ and $w =_\top v$ in
$\Lang(\catC)$. This is obtained by the definition of a model, the definition
of a $\mathcal{V}$-category, and the definition of $\Lang(\catC)$. This
observation allows to establish the surjective map that sends $[v]$ to $[v]$,
\ie it sends the equivalence class of $v$ as a $\lambda$-term in
$\Lang^\lambda(\catC)$ into the equivalence class of $v$ as a $\lambda$-term in
$\Lang(\catC)$.
\end{proof}

Finally we state,

\begin{thm}\label{theo:intV}
        Consider a $\VCatSe$-autonomous category $\catC$. There exists a
        $\VCatSe$-equivalence $\Syn(\Lang(\catC)) \simeq \catC$ and both
        functors witnessing this equivalence are autonomous. The functor going
        from the left to right direction is additionally strict. 
\end{thm}

\begin{proof}
  Let $a$ denote the underlying function of the hom-($\mathcal{V}$-categories)
  in $\Syn(\Lang(\catC))$ and $b$ the underlying function of the
  hom-($\mathcal{V}$-categories) in $\catC$.  We have, by construction, a model
  of $\Lang(\catC)$ on $\catC$ which acts as the identity in the interpretation
  of ground types and operation symbols. We can then appeal 
  to~\cref{th:strV} to establish a $\VCatSe$-functor $\Syn(\Lang(\catC))
  \to \catC$. Next, the functor working on the inverse direction behaves as the
  identity on objects and sends a morphism $f$ into $[f(x)]$. Let us show that
  it is $\VCatSe$-enriched. First, observe that if $q \ll b(f,g)$ in $\catC$
  and $q \in B$ then $f(x) =_q g(x)$ is a theorem of $\Lang(\catC)$, due to the
  fact that $\ll$ entails $\leq$ and by the definition of $\Lang(\catC)$. Using
  the definition of a basis, we thus obtain $b(f,g) = \bigvee \{ q \in B \mid
  \> q \ll b(f,g) \} \leq \bigvee \{ q \in B \mid f(x) =_q g(x) \} =
  a([f(x)],[g(x)])$.
  The fact that both functors are autonomous and that the one from left to right
  direction is additionally strict is proved as in~\cref{th:int}. The
  same applies to showing that the functors define an equivalence of categories.
\end{proof}

\begin{thm}\label{theo:fequiv}
        There exists an equivalence,
        \begin{align}
        \xymatrix@C=40pt{
                \V\lambda\text{-} \catfont{Th}
                \ar@<5pt>[r]^(0.40){\Syn}_(0.4){\simeq} & 
                \VCatSe\text{-}\Aut_{/_{\cong}}
                \ar@<5pt>[l]^(0.60){\Lang}
        }
        \end{align}
\end{thm}
\begin{proof}
        The functor $\Syn : \V \lambda\text{-}\catfont{Th} \to
        \VCatSe\text{-}\Aut_{/_{\cong}}$ sends a $\V\lambda$-theory to its
        syntactic category and acts as the identity on morphisms. For the
        inverse direction, recall that for every $\VCatSe$-autonomous
        category $\catC$, ~\cref{theo:intV} provides
        $(\VCatSe)$-autonomous functors $e_\catC : \catC \to
        \Syn(\Lang(\catC))$ and $e'_\catC : \Syn(\Lang(\catC)) \to \catC$. So
        then we define  $\Lang : \VCatSe\text{-}\Aut_{/_{\cong}}\to \V
        \lambda\text{-}\catfont{Th}$ as the functor that sends a
        $\VCatSe$-autonomous category $\catC$ to its theory $\Lang(\catC)$
        and that sends an isomorphism class $[\funF]$ of $(\VCatSe)$-autonomous
        functors of type $\catC \to \catD$ into $[ e_\catD \comp \funF \comp
        e'_\catC]$. This last functor is well-defined for the reasons already
        detailed in~\cref{th:int}. Similarly, the fact that both
        functors $\Syn$ and $\Lang$ form an equivalence follows a
        reasoning analogous to the one detailed in the proof of~\cref{th:eq}.
\end{proof}
\cref{theo:fequiv} and the preceding results extend straightforwardly to
symmetric linear $\V \lambda$-theories and $\VCatSS$-autonomous categories. As
before when $\mathcal{V}$ is the Boolean quantale we obtain the classical
equivalence theorem described in the previous section (\cref{th:eq}).

\section{Examples of linear $\V \lambda$-theories and their models}
\label{sec:examples}
\subsection{Real-time computation}

We now return to the example of wait calls and its  metric axioms~\eqref{ax}
sketched in Section~\ref{S:one}. Let us build a model on $\Met$ for this
theory.  Denote by $\Nats$ the metric space of natural numbers. Then fix a
metric space $A$, interpret the ground type $X$ as $\Nats \otimes A$ and the
operation symbol $\prog{wait_n} : X \to X$ as the non-expansive map, $
\sem{\prog{wait_n}} : \Nats \otimes A \to \Nats \otimes A$ , $(i,a) \mapsto (i
+ n, a)$.  Since we already know that $\Met$ is a $\Met$-autonomous category
(recall Definition~\ref{defn:enr_aut} and Example~\ref{ex:pos_met}) we only
need to show that the axioms in~\eqref{ax} are satisfied by the proposed
interpretation. This can be shown via a few routine calculations.  

As an illustration of this theory at work, let us use it to briefly explore
what would happen if one freely allowed multiple uses of the same variable in
building a judgement. Consider the \emph{non-linear} term $\lambda f, x.\
f(f(\dots(f \> x)))$ which we abbreviate to $\prog{seq^n}$. Intuitively, it
applies the operation given as argument $n$-times.  For example in
$\prog{seq^n}\ (\lambda x.\ \prog{wait_1}(x))\ z$ the operation $\lambda x.\
\prog{wait_1}(x)$ is applied $n$-times. If $\prog{seq^n}$ was allowed in our
calculus then we would obtain problematic metric theorems, such as,
\[
\prog{seq^n}\ (\lambda x.\ \prog{wait_1}(x))\ z =_1
\prog{seq^n}\ (\lambda x.\ \prog{wait_2}(x))\ z
\]
The theorem makes no sense because the left-hand side unfolds to an execution
time of $n$ seconds and the right-hand side unfolds to an execution time of
$2n$ seconds. In the linear setting, we can rewrite $\prog{seq^n}$ into the
more general term $\lambda f_1,\dots,f_n, x.\ f_1(f_2(\dots(f_n \> x)))$ and
then via our $\V$-deductive system obtain,
\[
\prog{seq^n}\ (\lambda x.\ \prog{wait_1}(x))\ \dots\
(\lambda x.\ \prog{wait_1}(x))\ z =_{n \cdot 1}
\prog{seq^n}\ (\lambda x.\ \prog{wait_2}(x))\ \dots\ 
(\lambda x.\ \prog{wait_2}(x))\ z
\]
which is in line with the total execution times of $n$ and $2n$ seconds
mentioned above. This brief exploration tells that it is important to track
the use of resources in metric $\lambda$-theories.

  Now, it may be the case that is unnecessary to know the \emph{distance}
  between the execution time of two programs -- instead it suffices to know
  whether a program finishes its execution \emph{before} another one.  This
  leads us to linear $\mathcal{V}\lambda$-theories where $\mathcal{V}$ is the
  Boolean quantale. We call such theories \emph{linear ordered
  $\lambda$-theories}.  Recall again the language from the introduction with a single
  ground type $X$ and the signature of wait calls $\Sigma = \{ \prog{wait_n} :
  X \to X \mid n \in \Nats \}$. Then we adapt the metric axioms \eqref{ax} to
  the case of the Boolean quantale by considering instead:
  \begin{flalign*}
    \prog{wait_0}(x) = x \hspace{1cm} \prog{wait_n}(\prog{wait_m}(x)) =
    \prog{wait_{n + m }}(x) \hspace{1cm}
    \infer{\prog{wait_n}(x) \leq \prog{wait_m}(x)}{n \leq m}
  \end{flalign*}
  where a classical equation $v = w$ is shorthand for $v \leq w$ (\ie $v =_1
  w$) and $w \leq v$ (\ie $w =_1 v$).  In the resulting theory we can consider
  for instance (and omitting types for simplicity) the $\lambda$-term that
  defines the composition of two functions $\lambda f.\ \lambda g.\ g\ (f\ z)$,
  which we denote by $v$, and show that,
  \begin{flalign*} 
    v\ (\lambda x.\ \prog{wait_1}(x)) \leq v\
    (\lambda x.\ \prog{wait_1}(\prog{wait_1}(x)))
  \end{flalign*}
  This inequation between higher-order programs arises from the argument
  $\lambda x.\ \prog{wait_1}(\prog{wait_1}(x))$ being costlier than the
  argument $\lambda x.\ \prog{wait_1}(x)$ -- specifically the former invokes
  one more wait call ($\prog{wait_1}$) than the latter. Intuitively, the
  inequation entails that for every argument $g$ the execution time of
  computation $v\ (\lambda x.\ \prog{wait_1}(x))\ g$ will always be smaller
  than that of computation $v\ (\lambda x.\ \prog{wait_1}(\prog{wait_1}(x)))\
  g$ since it invokes one more wait call. In general the inequation reflects the fact
  that costlier programs fed as input will result in longer execution
  times when performing the corresponding computation.

  In order to build a model for the ordered theory of wait calls, let now
  $\Nats$ be the poset of natural numbers. Then fix a poset $A$ and define a
  model over $\Pos$ by sending $X$ into $\Nats \otimes A$ and $\prog{wait_n} :
  X \to X$ to the monotone map $ \sem{\prog{wait_n}} : \Nats \otimes A \to
  \Nats \otimes A$, $(i,a) \mapsto (i + n, a)$.  Since we already know that
  $\Pos$ is a $\Pos$-autonomous category (recall Definition~\ref{defn:enr_aut}
  and Example~\ref{ex:pos_met}) we only need to show that the ordered axioms
  are satisfied by the proposed interpretation.  But again, this can be shown
  via a few routine calculations.

\subsection{Probabilistic computation}
\label{sec:prob}
Let us now analyse an example of a metric $\lambda$-theory concerning
probabilistic computation. We start by considering ground types $\real,\real^+$
and $\unit$, and a signature of operations consisting of $\{ r : \typeI \to
\real \mid r \in \mathbb{Q} \} \cup \{ r^+ : \typeI \to \real^+ \mid r \in
\mathbb{Q}_{\geq 0}\} \cup \{ r^u : \typeI \to \unit \mid r \in [0,1] \cap
\mathbb{Q} \}$, an operation $+$ of type $\real,\real\to\real$, and
\emph{sampling} functions $\bern: \real,\real,\unit\to\real$ and
$\normal:\real,\real^+\to\real$. Whenever no ambiguities arise, we drop the
superscripts in $r^u$ and $r^+$.  Operationally, $\bern(x,y,p)$ generates a
sample from the Bernoulli distribution with parameter $p$ on the set $\{x,y\}$,
whilst $\normal(x,y)$ generates a normal deviate with mean $x$ and standard
deviation $y$.  We then postulate the metric axiom,
\begin{flalign}
\infer{\bern(x_1,x_2,p(\ast)) =_{\abs{p-q}} \bern(x_1,x_2,q(\ast))}
	{p, q\in [0,1] \cap \mathbb{Q}} \label{eq:bern}
\end{flalign}
Consider now the following $\lambda$-terms (where we abbreviate the constants
$0(\ast),1(\ast),p(\ast),q(\ast)$ to $0,1,p,q$, respectively),
\begin{flalign*}
& \mathtt{walk1}\defeq \lambda x:\real. \bern(0,x+\normal(0,1),p)
\\
& \mathtt{walk2}\defeq \lambda x:\real. \bern(0,x+\normal(0,1),q), \qquad p,q\in[0,1] \cap \mathbb{Q}.
\end{flalign*}
As the names suggest, these two terms of type $\real\multimap\real$ correspond
to random walks on $\Reals$. At each call, $\mathtt{walk1}$ (\resp
$\mathtt{walk2}$) performs a jump drawn randomly from a standard normal
distribution, or is forced to return to the origin with probability $p$ (\resp
$q$).  These are non-standard random walks whose semantics tend to be given by
complicated operators (details below), but the simple $\V$-equational system of
Fig.~\ref{fig:theo_rules} and the axiom \eqref{eq:bern} allow us to easily
derive $\mathtt{walk1} = _{\abs{p-q}}\mathtt{walk2}$ without having to compute
the semantics of these terms. In other words, the axiom \eqref{eq:bern} is
enough to tightly bound the distance between two non-trivial random walks
represented as higher-order terms in a probabilistic programming language.
Furthermore the tensor in $\lambda$-calculus allows to easily scale up this
reasoning to random walks on higher dimensions such as $\lambda x :
\prog{real \otimes real}.\ \prog{pm}\ x\ \prog{to}\ x_1 \otimes x_2.\
\mathtt{walk1}(x_1)\otimes\mathtt{walk2}(x_2) : (\real \otimes \real) \multimap
(\real \otimes \real)$ on $\Reals^2$.

We now interpret the resulting metric $\lambda$-theory in the category
$\Ban$ of Banach spaces and short operators, \ie the semantics of
\cite{DK20a,K81c} without the order structure needed to interpret
\texttt{while} loops.  This is the usual representation of Markov
chains/kernels as matrices/operators. Recall that every operator
$T : V \to U$ between Banach spaces $V$ and $U$ has a norm $\norm{T}$,
\[
        \norm{T} = \bigvee\ \{\ \norm{T v} \mid \norm{v} = 1 \}
\]
called \emph{operator norm}. Recall also that short operators satisfy the
inequation $\norm{T} \leq 1$. It is well known that $\Ban$ has an autonomous
structure~\cite{DK20a,K81c} where the tensor product is the projective tensor
$\hat \otimes_\pi$. So we will just focus on showing that $\Ban$ is a
$\Met$-enriched autonomous category (\ie an instance of
Definition~\ref{defn:enr_aut}). First, recall that a norm on a vector space
induces a metric~\cite{ryan2013introduction}, in particular we obtain a metric
$d$ on the hom-set $\Ban(V,U)$ which is defined by,
\[
        d(S,T)= \norm{S - T}
\]
for $S,T \in \Ban(V,U)$. For all operators $T,S$ between Banach spaces, it is
also well-known that if $S$ is a short operator then $\norm{S T} \leq
\norm{T}$, and if $T$ is a short operator then $\norm{S T} \leq \norm{S}$.
Moreover it is the case that $\norm{T \hat \otimes_\pi \id} \leq \norm{T}$ and
$\norm{\id \hat \otimes_\pi T} \leq \norm{T}$ (see \cite[\S
2.1]{ryan2013introduction}). From these results it follows that,

\begin{prop}
        \label{prop:ban}
        The category $\Ban$ is $\Met$-enriched and the bifunctor
        $\hat \otimes_\pi : \Ban \otimes \Ban \to \Ban$ is 
        $\Met$-enriched as well.
\end{prop}

\begin{proof}
        Let us first show that $\Ban$ is $\Met$-enriched. We deduce by
        unfolding the respective definitions that we
        need to show the following: for all short operators $T,T' : V \to U$ and $S,S' :
        U \to W$ the inequation $\norm{S - S'} + \norm{T - T'} \geq \norm{ST -
        S'T'}$ holds. So we reason in the following way:
        \begin{flalign*}
              &  \, \norm{S - S'} + \norm{T - T'} & \\
              & \geq \norm{(S - S')T} + \norm{S'(T-T')} 
              & \left \{\norm{(S- S')T} \leq \norm{S - S'} 
              \text{ and } \norm{S'(T-T')} \leq \norm{T - T'} \right \} \\
              & = \norm{ST - S'T} + \norm{S'T - S'T'} & \\
              & \geq \norm{ST - S'T + S'T - S'T'} & \text{\{Triangle inequality\}} \\
              & = \norm{ST - S'T'}
        \end{flalign*}
        Next, concerning $\hat \otimes_\pi$ we can also deduce by unfolding the
        respective definitions that we need to prove $\norm{S - S'} +
        \norm{T - T'} \geq \norm{T \hat \otimes_\pi S - T' \hat \otimes_\pi
        S'}$. For this case we calculate,
        \begin{flalign*}
              &  \, \norm{S - S'} + \norm{T - T'} & \\
              & \geq \norm{\id \hat \otimes_\pi (S - S')} 
              + \norm{(T-T') \hat \otimes_\pi \id } 
              & \Big \{ \norm{\id \hat \otimes_\pi (S- S')} \leq \norm{S - S'} 
               \text{ and }\dots \\
              & &  \norm{(T-T') \hat \otimes_\pi \id} 
                \leq \norm{T - T'} \Big \}  \\
              & = \norm{\id \hat \otimes_\pi S - \id \hat \otimes_\pi S'} 
              + \norm{T \hat \otimes_\pi \id - T' \hat \otimes_\pi \id} & \\
              & \geq \norm{(\id \hat \otimes_\pi S) \comp (T \hat \otimes_\pi \id)
                - (\id \hat \otimes_\pi S') \comp (T' \hat \otimes_\pi \id)} & 
              \text{\{$\Ban$ is $\Met$-enriched\}} \\
              & = \norm{T \hat \otimes_\pi S - T' \hat \otimes_\pi S'}
              & \qedhere
        \end{flalign*}
 
\end{proof}

We will next recur to the following general theorem for showing that $\Ban$ is
a $\Met$-enriched autonomous category.

\begin{thm}
        \label{thm:aut}
        Consider an autonomous category $\catC$ such that,
        \begin{enumerate}
                \item $\catC$ is $\VCat$-enriched; 

                \item the functor $\otimes : \catC \times \catC \to \catC$ is
                        $\VCat$-enriched;
                
                \item and for all $\catC$-objects $X$ the functor $(X \multimap
                        -) : \catC \to \catC$ is $\VCat$-enriched as well.  
        \end{enumerate}
        Then $\catC$ is a $\VCat$-enriched autonomous category.
\end{thm}

\begin{proof}
        We first show that for all $\catC$-objects $X$ the functor $(- \otimes
        X) : \catC(A, B) \to \catC(A \otimes X, B \otimes X)$ defined by $h
        \mapsto h \otimes \id$ is $\VCat$-enriched.  This follows directly
        from (2) and from the fact that the corresponding mapping is given by
        the composite,
        \[
                \xymatrix@C=38pt{
                        \catC(A,B) \cong \catC(A,B) \otimes 1 
                        \ar[r]^(0.53){\id \otimes\, (\ast \mapsto \id_X)} & 
                        \catC(A,B) \otimes \catC(X,X) \ar[r]^(0.52){\otimes} & 
                        \catC(A \otimes X, B \otimes X)
                }
        \]
        We thus obtain that for all $\catC$-objects $X$ both functors $(-
        \otimes X)$ and $(X \multimap -)$ are $(\VCat)$-enriched. The next step
        is to prove that the maps witnessing the isomorphism $\catC(Y \otimes
        X, Z) \cong \catC(Y, X \multimap Z)$ are $\V$-functors. The
        left-to-right direction is the composite,
        \begin{align*}
        &       \xymatrix@C=40pt{
                        \catC(Y \otimes X, Z)
                        \ar[r]^(0.24){(X \multimap -)} & 
                        \catC(X \multimap (Y \otimes X), X \multimap Z) 
                        \cong 1 \otimes \catC(X \multimap (Y \otimes X), 
                        X \multimap Z) 
                                        &
                } \\[-2pt]
        &       
                \xymatrix@C=40pt{
                        \ar[r]^(0.10){(\ast \mapsto \overline{ \id_{Y \otimes X}})\, 
                        \otimes\, \id} 
                        & 
                        \catC(Y, X \multimap (Y \otimes X)) 
                                \otimes \catC(X \multimap (Y \otimes X), 
                        X \multimap Z) \ar[r]^(0.72){(\cdot)} &
                        \catC(Y, X \multimap Z)
                }
        \end{align*}
        and so by $(3)$ and $(1)$  it must be a $\V$-functor.
        The right-to-left direction is the composite,
        \begin{align*}
        &       \xymatrix@C=40pt{
                        \catC(Y, X \multimap Z) \ar[r]^(0.26){(- \otimes X)} &
                        \catC(Y \otimes X, (X \multimap Z) \otimes X) \cong 
                        \catC(Y \otimes X, (X \multimap Z) \otimes X) \otimes 1 
                        \ar[r]^(0.88){\id \otimes (\ast \mapsto \app) } &
                } \\[-2pt]
        &       \xymatrix@C=40pt{
                        \catC(Y \otimes X, (X \multimap Z) \otimes X) \otimes 
                        \catC( (X \multimap Z) \otimes X, Z) \ar[r]^(0.72){(\cdot)} &
                        \catC(Y \otimes X, Z)
                }
        \end{align*}
        Since $(- \otimes X)$ is $\VCat$-enriched, and $\catC$ is autonomous
        and $\VCat$-enriched, this composite must be a $\V$-functor as well.
\end{proof}
The previous theorem can easily be adjusted to work with the bases of
enrichment $\VCatSe$ and $\VCatSS$. We then obtain,

\begin{thm}
        The category $\Ban$ is a $\Met$-enriched autonomous category.
\end{thm}

\begin{proof}
        Proposition~\ref{prop:ban} and~\cref{thm:aut} entail our claim
        if we show that the functor $(V \multimap -) : \Ban \to \Ban$ is
        $\Met$-enriched for all Banach spaces $V$. Thus recall that for all
        Banach spaces $W$ and $U$ the mapping $(V \multimap -) : \Ban(W,U) \to
        \Ban(V \multimap W, V \multimap U)$ is defined as $S \mapsto (S \comp
        -)$. Then to prove that $(V \multimap -)$ is $\Met$-enriched recall as
        well that the space $V \multimap W$ is equipped with the operator norm,
        and that all elements $T$ of this space with $\norm{T} = 1$ are
        necessarily short. Finally observe that for all $S \in \Ban(W,U)$ we have
        $\norm{S} \geq \norm{(S \comp -)}$, because for all operators $T \in V
        \multimap W$ with $\norm{T} = 1$ the inequation $\norm{S} \geq \norm{S
        \comp T}$ holds. It follows that the mapping $(V \multimap -)
        : \Ban(W,U) \to \Ban(V \multimap W, V \multimap U)$, $S \mapsto (S
        \comp -)$ is non-expansive.
\end{proof}        

Finally, $\Ban$ forms a model of our metric $\lambda$-theory via the following
interpretation:  we define $\sem{\real}=\meas\Reals$, the Banach space of
finite Borel measures on $\Reals$ equipped with the total variation norm, and
similarly $\sem{\real^+}=\meas\Reals^+$ and $\sem{\unit}=\meas[0,1]$.  We have
$\sem{\typeI}=\Reals \ni 1$, and for every $r\in\mathbb{Q}$ we put
$\sem{r}:\Reals\to\meas\Reals,x\mapsto x\delta_r$, where $\delta_r$ is the
Dirac delta over $r$; thus $\sem{r}(1)=\delta_r$. We define an analogous
interpretation for the operation symbols $r^+$ and $r^u$.  For $\mu, \upsilon
\in\meas\Reals$ we define $\sem{+}(\mu \otimes \upsilon)=
+_{\ast}(\mu\otimes\upsilon)$, \ie the pushforward under $+$ of the product
measure $\mu \otimes \upsilon$ (seen as an element of
$\meas(\Reals \times\Reals)$, see \cite{DK20a}). For $\mu,\upsilon,\xi
\in\meas\Reals$ we define $\sem{\bern} (\mu \otimes \upsilon \otimes \xi)=
\mathrm{bern}_{\ast}(\mu \otimes \upsilon \otimes \xi)$, the pushforward of the
product measure $\mu \otimes \upsilon \otimes \xi$ under the Markov kernel
$\mathrm{bern}: \Reals^3\to \Reals, (u,v,p)\mapsto p\delta_u + (1-p)\delta_v$,
and similarly for $\sem{\normal}$ (see \cite{DK20a} for the definition of
pushforward and pushforward by a Markov kernel). This interpretation
is sound because the norm on $\meas\Reals$ is the total variation norm, and the
metric axiom~\eqref{eq:bern} describes the total variation distance between the
corresponding Bernoulli distributions.  More formally, for
$\sem{\bern(x,y,p)}$ and $\sem{\bern(x,y,q)}$ with
$p,q\in[0,1]$ we calculate, 
\begin{flalign*}
&\bigvee_A \abs{\int\hspace{-4pt}\int \hspace{-1pt} p\delta_u(A)\hspace{-2pt} +\hspace{-2pt} (1\hspace{-2pt}-\hspace{-2pt}p)\delta_v(A) \hspace{-1pt} d\sem{x}(du)d\sem{y}(dv)\hspace{-2pt} - \hspace{-3pt}\int\hspace{-4pt}\int \hspace{-1pt} q\delta_u(A) \hspace{-2pt}+\hspace{-2pt}(1\hspace{-2pt}-\hspace{-2pt}q)\delta_v(A) \hspace{-1pt} d\sem{x}(du)d\sem{y}(dv)}
\\
=&\bigvee_A\abs{\int\hspace{-4pt}\int (p-q)\delta_u(A) + ((1-p)-(1-q))\delta_v(A)~d\sem{x}(du)d\sem{y}(dv)}
\\
=&\bigvee_A\abs{(p-q)\sem{x}(A) + (q-p)\sem{y}(A)}
\\
=&\bigvee_A\abs{(p-q)(\sem{x}(A) -\sem{y}(A))}
\\
\leq & ~\abs{p-q}
\end{flalign*}

\subsection{Quantum computation}

We now turn our attention to quantum computing, a paradigm in which notions
of approximation and noise take a central role~\cite[Chapter 9]{nielsen02}.
Since quantum computing is perhaps less known than the probabilistic case, we
start by recalling some basic concepts on the topic -- we assume however some
familiarity with linear algebra. After these preliminaries we present a very
simple metric theory based on the idea of approximating a quantum operation
(given in terms of a rotation). We then start studying possible models for
this theory. An obvious candidate is the category of isometries but we will
see that it unfortunately carries a too fine-grained metric for quantum
computing.  We will thus move to a category of so-called quantum channels
that fixes this problem. This latter category is not monoidal closed so we
will take the $\Met$-enriched version of its presheaf category as
interpretation domain.

Recall that for every natural number $n \geq 1$ the set $\C^n$ is a complex
vector space, in fact an inner product space with the inner product $\langle v,
u \rangle$ given by $v^\dagger u$ where $(-)^\dagger$ is the adjoint operation.
Recall as well that inner products induce a norm $\norm{v}= \sqrt{\langle v, v
\rangle}$. A quantum state is an element $v \in \C^n$ such that $\norm{v} = 1$.
We denote the elements of the canonical basis of $\C^n$ by
$\ket{0},\dots,\ket{n-1}$.  Consider now a linear map $T : \C^n \to \C^m$. If
$\norm{Tv} = \norm{v}$ for all elements $v \in \C^n$ we call $T$ an isometry.
Note that isometries are precisely those linear maps that send quantum states to
quantum states, so it makes sense to focus on
isometries (at least for now). An important isometry is the so-called
\emph{phase} operation $P(\phi) : \C^2 \to \C^2$ ($\phi$ is an angle in
radians) which is  defined by $\ket{0} \mapsto \ket{0}$, $\ket{1} \mapsto
e^{i\phi}\ket{1}$. This operation can be elegantly explained geometrically in
the following way: first, it is well-known that when global phases are ignored
we can represent a quantum state $v \in \C^2$ in the form,
\[
        \cos \frac{\theta}{2} \ket{0} + e^{i \varphi} \sin \frac{\theta}{2}\ket{1}
\]
which corresponds to a point in the unit sphere where $\theta$ marks the
latitude (\ie the polar angle) and $\varphi$ marks the longitude (\ie the
azimuthal angle). This representation is traditionally called the \emph{Bloch
sphere representation}.  A point in the latter representation corresponds to
the vector in $\mathbb{R}^3$ defined by $(\cos \varphi \sin \theta, \sin
\varphi \sin \theta, \cos \theta)$ and often called Bloch vector~\cite[page
174]{nielsen02}. Then via routine linear algebra calculations we see that
$P(\phi)$ corresponds to a rotation of $\phi$ radians along the $z$-axis, \ie
the polar angle is preserved and the azimuthal angle changes to $\varphi +
\phi$. Another important isometry is $R_y(\phi) : \C^2 \to \C^2$ which
corresponds to a rotation of $\phi$ radians along the $y$-axis and is given by
the matrix,
\[
        \begin{pmatrix}
                \cos\ \frac{\phi}{2} & - \sin\ \frac{\phi}{2} \\
                \sin\ \frac{\phi}{2} & \cos\ \frac{\phi}{2}
        \end{pmatrix}
\]
An isometry often used in quantum computing is the \emph{Hadamard gate}
$H : \C^2 \to \C^2$ which is defined by $H = R_y(\frac{\pi}{2}) P(\pi)$.

\noindent\textbf{Quantum random walk on a (discrete) circle.} 
Quantum random walks have been extensively studied in the quantum
literature~\cite{venegas12}. They are part of several quantum algorithms, and
although conceptually similar to the probabilistic counterpart they tend to
exhibit very different behaviours. Here we focus on a simple
case of a quantum random walk and use it to provide another example of a metric
$\lambda$-theory.

Let us start by defining a metric $\lambda$-theory whose ground types are
$\prog{qbit}$ and $\prog{pos}$. In this example values of type $\prog{qbit}$
will be used to decide whether to move left or right, and $\prog{pos}$ will be
used to mark the current position on a discrete circle. As operations we have
$\prog{H} : \prog{qbit} \to \prog{qbit}$ (which will later be interpreted as
the Hadamard gate) and $\prog{S} : \prog{qbit} \otimes \prog{pos} \to
\prog{qbit} \otimes \prog{pos}$, which, as alluded above, will use a value of
type $\prog{qbit}$ to move either to the left or right. Here one step of the
quantum walk is given by,
\[
        \lambda x: \prog{qbit} \otimes \prog{pos}.\ \prog{pm}\ x\
        \prog{to}\ x_1 \otimes x_2.\ \prog{S}(\prog{H}(x_1) \otimes x_2) :
        \prog{qbit}\otimes\prog{pos} \multimap
        \prog{qbit}\otimes\prog{pos} 
\]
which we abbreviate to $\prog{step} : \prog{qbit}\otimes\prog{pos} \multimap
\prog{qbit}\otimes\prog{pos}$. If we wish to perform $n$-steps, then we can
proceed modularly in the following way. First we define the closed judgement
$\lambda f_1,\dots,f_n, x.\ f_1(f_2(\dots(f_n \> x)))$ which we abbreviate to
$\prog{seq^n}$ (recall the example of wait calls above). Then we can
straightforwardly build the closed judgement $\prog{seq^n}\ \prog{step}\ \dots\
\prog{step} : \prog{qbit} \otimes \prog{pos} \multimap \prog{qbit} \otimes
\prog{pos}$ to represent $n$-steps in the walk.

Let us now consider a simple metric axiom: we add a new operation
$\prog{H}^\epsilon : \prog{qbit} \to \prog{qbit}$ and postulate as axiom that
$x : \prog{qbit} \vljud \prog{H}(x) =_\epsilon \prog{H}^\epsilon(x) :
\prog{qbit}$.  The operation $\prog{H}^\epsilon$ represents an \emph{imperfect}
implementation of $\prog{H}$. For example by knowing that the Hadamard
gate is the composition $R_y(\frac{\pi}{2}) \comp P(\pi)$ we may regard
$\prog{H}^\epsilon$ as the composition $R_y(\frac{\pi}{2})\comp P(\pi +
\delta)$ which does not rotate along the $z$-axis precisely $\pi$ radians but
$\pi + \epsilon$ radians instead. This kind of imperfection is unavoidable in
the implementation of quantum gates. Next, we denote by
$\prog{step}^\epsilon$ the judgement that results from replacing $\prog{H}$ in
$\prog{step}$ by $\prog{H}^\epsilon$, and can easily obtain $\prog{step}
=_\epsilon \prog{step}^\epsilon$ via our metric deductive system. Finally we
can deduce,
\[
        \prog{seq^n}\ \prog{step}\ \dots\
        \prog{step} =_{n \cdot \epsilon}
        \prog{seq^n}\ \prog{step}^\epsilon\ \dots\ \prog{step}^\epsilon 
\]
This last metric equation gives the important message that by closing the
distance between $\prog{H}$ and $\prog{H}^\epsilon$ (with $\prog{H}^\epsilon$
corresponding to an imperfect implementation of $\prog{H}$) we can also close
the distance between an idealised quantum walk and its implementation when the
walk is bounded by a specific number of steps. This is of course necessary to
be able to \emph{actually execute} a quantum walk that is close to our
idealisation of it. 

Let us now build a model for the theory just introduced. An obvious
candidate for the interpretation domain is the category $\Iso$ whose objects
are natural numbers $n \geq 1$ and morphisms $n \to m$ are isometries $\C^n \to
\C^m$. This category is strict symmetric monoidal with the tensor given by the
usual tensor product of vector spaces and symmetry given by the isometry $v
\otimes w \mapsto w \otimes v$~\cite{staton18}. Moreover the set of
isometries $\Iso(n,m)$ can be equipped with the metric induced by the
\emph{operator norm} which is defined as,
\[
        \norm{T} = \bigvee \left \{ \norm{T v} \mid 
        \norm{v} = 1 \right  \}
\]
for $T : n \to m$ an isometry. Unfortunately, the category $\Iso$ has two
important shortcomings:
\begin{enumerate}
        \item it is not monoidal closed -- and although higher-order structure
                is not frequently used in quantum algorithmics we know that it
                renders program constructions and deductions more modular (as
                illustrated with the quantum walk above and $\prog{seq^n}$);
        \label{i:ho}
        \item the metric induced by the operator norm is too fine-grained for
                quantum computing. Indeed quantum states should be
                \emph{indistinguishable up to global phase}, but even so we can
                define two isometries $T,S : 1 \to 2$ such that $T 1 =
                \ket{0}$, $S 1 = -\ket{0}$ and the distance between $\norm{T-S}$
                will be at least $2$.
        \label{i:metric}
\end{enumerate}
So we now recall a formalism for quantum computing that produces a
category that fixes item \eqref{i:metric} (item~\eqref{i:ho} will be handled
later on). We will need some preliminaries: a matrix $A\in\C^{n\times n}$ is
said to be \emph{positive semi-definite} (or just positive), notation $A\geq
0$, if $\langle v, A v \rangle \geq 0$ for all vectors $v \in \C^n$. We denote the trace
of a matrix $A$ by $\Tr\ A$. A matrix $A\in\C^{n\times n}$ is
\emph{Hermitian} if $A=A^\dagger$.  Positive matrices are Hermitian~\cite[page
71]{nielsen02}. A positive matrix $A \in \C^{n \times n}$ with $\Tr\ A=1$ is
called a \emph{density matrix} or a \emph{density operator}. A mixed quantum
state is a convex combination $p_1 \cdot v_1v_1^\dagger + \dots + p_n \cdot
v_nv^\dagger_n$ of quantum states $v_1,\dots,v_n \in \C^n$ and such a
combination corresponds precisely to a density matrix~\cite{nielsen02}. One
usually denotes density matrices by the greek letters $\rho,\sigma$ and so
forth. A density matrix encodes uncertainty about the current state of the
quantum system at hand. For example, $\frac{1}{2} \cdot \ket{0}\ket{0}^\dagger
+ \frac{1}{2}\cdot \ket{1}\ket{1}^\dagger$ tells that the current state is
either $\ket{0}$ or $\ket{1}$ with probability $\frac{1}{2}$. Note that global
phases disappear in this formalism because given a quantum state $e^{i \theta}
v$ we have $(e^{i \theta} v)(e^{i \theta} v)^\dagger = e^{i
\theta}\overline{e^{i \theta}} v v^\dagger = v v^\dagger$.

An operator between spaces of matrices $T: \C^{n\times n}\to \C^{m\times m}$ is
often called a super-operator.  A super-operator $T: \C^{n\times n}\to
\C^{m\times m}$ is called \emph{positive} if it sends positive matrices to
positive matrices, \ie $A\geq 0\Rightarrow TA\geq 0$.  Since density matrices
are positive,  it is clearly necessary that a physically allowed transformation
be represented by a positive operator.  In fact more is needed: since one can
always extend the space $\C^{n\times n}$ to a space $\C^{n\times n}\otimes
\C^{k\times k}$ by adjoining a new quantum system, any physically allowed
transformation $T: \C^{n\times n}\to \C^{m\times m}$ must have the property
that $T\otimes \id_{\C^{k\times k}}:\C^{n\times n}\otimes \C^{k\times k}\to
\C^{m\times m}\otimes \C^{k\times k}$ is positive.  A super-operator $T$
satisfying this condition is called \emph{completely positive}. Finally, a
super-operator $T$ is called \emph{trace-preserving} if $\Tr\ TA=\Tr\ A$.
Completely positive, trace-preserving super-operators are traditionally called
\emph{quantum channels}. 

It is straightforward to prove that quantum channels are closed under operator
composition and tensoring~\cite{selinger2004towards}. In fact, we have the
strict symmetric monoidal category $\CPTP$ whose objects are natural numbers $n
\geq 1$ and morphisms $n \to m$ are quantum channels $\C^{n \times n} \to \C^{m
\times m}$.  Note also the existence of a strict symmetric monoidal functor $r:
\Iso \to \CPTP$ (formally a reflection) that behaves as the identity on objects
and that sends an isometry $T$ to the mapping $A \mapsto T A
T^\dagger$~\cite{staton18}.

Our next step is to prove that $\CPTP$ is a $\Met$-enriched symmetric monoidal
category.  For that effect we will require more preliminaries: first a matrix
$A\in\C^{n\times n}$ is said to be \emph{normal} if $AA^\dagger=A^\dagger A$.
Clearly every Hermitian matrix is normal.  Note also that for every matrix $A
\in \C^{n \times n}$ the matrix $A^\dagger A$ is Hermitian.  Next, it is
well-known that by appealing to the spectral theorem~\cite{nielsen02}, every
normal matrix $A \in \C^{n \times n}$ can be expressed as a linear combination
$\sum_i \lambda_i b_i b^\dagger_i$ where the set $\{ b_1, \dots, b_n \}$ is an
orthonormal basis of $\C^n$. Using this last result we can extend any function
$f: \C\to\C$ to normal matrices via,
\[
        f(A) = \sum_i f(\lambda_i) b_i b_i^\dagger
\]
We then obtain the norm $\norm{A}_1 = \Tr \sqrt{A^\dagger A}$ for matrices $A
\in \C^{n \times n}$. This norm is called the \emph{trace norm} and is also
known as the Schatten $1$-norm~\cite{watrous18}. The trace norm induces a
metric on the set of density matrices which is defined by $d(\rho,\sigma) =
\norm { \rho - \sigma}_1$. In the sequel we will often treat $vv^\dagger$ as if
it were simply $v$. Next, it is well known that the distance $d(vv^\dagger,
uu^\dagger)$ between two quantum states $v$ and $u$ is their \emph{Euclidean
distance} in the Bloch sphere~\cite{watrous18,nielsen02}. It is thus easy to
see for example that the distance between two quantum states,
\[
        \cos \frac{\theta}{2} \ket{0} + e^{i \varphi} \sin \frac{\theta}{2}\ket{1}
        \text{ and }
        \cos \frac{\theta}{2} \ket{0} + e^{i (\varphi + \epsilon)} 
        \sin \frac{\theta}{2}\ket{1}
\]
tends to $0$ when $\epsilon$ approaches $0$, more formally if $\epsilon_n \to
0$ then $f(\epsilon_n) \to 0$ where,
\begin{align*}
        f(\epsilon) & = \norm{ 
        (\cos \varphi \sin \theta, \sin \varphi \sin \theta, \cos \theta) -
        (\cos (\varphi + \epsilon) \sin \theta, \sin (\varphi + \epsilon) 
        \sin \theta, \cos \theta)  }_2 \\
        & = \norm{
        ( (\cos \varphi - \cos(\varphi + \epsilon)) \sin \theta,
          (\sin \varphi - \sin(\varphi + \epsilon)) \sin \theta,
          0 )
        }_2 \\
        & = \lvert \sin \theta \rvert 
        \norm{(\cos \varphi - \cos(\varphi + \epsilon),
          \sin \varphi - \sin(\varphi + \epsilon),
        0 )}_2
\end{align*}
Using results from topology, we know that this holds because $f(0) = 0$ and
moreover $f$ is continuous.  The definition of $f$ also tells us that for a
fixed $\epsilon$ the distance between the two states is maximised when $\lvert
\sin\ \theta \rvert = 1$ which corresponds to $\theta = \pm \frac{\pi}{2}$
(both states are located at the equator) and minimised when $\lvert \sin\
\theta \rvert = 0$ which corresponds to $\theta = 0$ or $\theta = \pi$ (both states are located at one of the poles which renders the azimuthal angle
irrelevant). We now consider the following norm on super-operators $T :
\C^{n\times n} \to \C^{m \times m}$~\cite{watrous18}:
\[
        \norm{T}_1 = \max\ \{ \norm{T A}_1 \mid  
        \norm{A}_1 = 1 \}
\]
Unfortunately, the norm $\norm{-}_1$ on super-operators is not stable under
tensoring~\cite{watrous18}, specifically the inequation $\norm{T \otimes \id}_1
\leq \norm{T}_1$ does not hold which makes impossible to enrich the monoidal
structure of $\CPTP$ via this norm. So instead we will use the so-called
\emph{diamond norm}, which given a super-operator $\C^{n\times n} \to \C^{m
\times m}$ is defined by,
\[
        \norm{T}_\diamond = \norm{T \otimes \id_n}_1
\]
Then it follows from \cite[Proposition 3.44 and Proposition 3.48]{watrous18}
that for all super-operators $T : \C^{n \times n} \to \C^{m \times m}$, $S :
\C^{m \times m}  \to \C^{o \times o}$ if $T$ is a quantum channel the
inequation $\norm{S T}_\diamond \leq \norm{S}_\diamond$ holds and if $S$ is a
quantum channel the inequation $\norm{S T}_\diamond \leq \norm{T}_\diamond$
also holds.  Furthermore, via \cite[Corollary 3.47]{watrous18} we have both
inequations $\norm{T}_\diamond \geq \norm{T \otimes \id}_\diamond$ and
$\norm{T}_\diamond \geq \norm{\id \otimes T}_\diamond$ for $T$ a
super-operator. As we saw previously with the category $\Ban$ these are
sufficient conditions to prove that $\CPTP$ is a $\Met$-enriched symmetric
monoidal category (recall Proposition~\ref{prop:ban}).

The exact calculation of distances induced by $\norm{-}_\diamond$ tends to be
quite complicated, but a useful property for calculating the distance
between quantum channels in the image of $r : \Iso \to \CPTP$ is provided in
\cite[Theorem 3.55]{watrous18}:
\begin{thm}
        \label{theo:smaller}
        Consider two isometries $T,S : n \to m$. There exists a unit
        vector $v \in \C^n$ such that,
        \[
                \norm{r(T)(vv^\dagger) - r(S)(vv^\dagger)}_1 =
                \norm{r(T) - r(S)}_\diamond
        \]
\end{thm}
As an illustration of the theorem at work, note that every isometry $1 \to n$
corresponds to the initialisation of a quantum state. So for two such
isometries $T$ and $S$ we can immediately see that the distance $\norm{r(T) -
r(S)}_\diamond$ between quantum channels $r(T)$ and $r(S)$ of type $1 \to n$ is
precisely $\norm{r(T)(1 1^\dagger) - r(S)(1 1^\dagger)}_1$.   For example the
distance between the $r$-image of the mapping $1 \mapsto \cos \frac{\theta}{2}
\ket{0} + e^{i \varphi} \sin \frac{\theta}{2}\ket{1}$ and the mapping $1
\mapsto \cos \frac{\theta}{2} \ket{0} + e^{i (\varphi + \epsilon) } \sin
\frac{\theta}{2}\ket{1}$ must be $f(\epsilon) = \lvert \sin \theta \rvert
\norm{(\cos \varphi - \cos(\varphi + \epsilon), (\sin \varphi - \sin(\varphi +
\epsilon), 0 )}_2$ which we already know tends to $0$ when $\epsilon$
approaches $0$. Let us now consider a more complex example, which is closely
related to the example above of a quantum random walk.

We already know that the phase operation $P(\phi) : \C^2 \to \C^2$ is an
isometry defined by $\ket{0} \mapsto \ket{0}$ and $\ket{1} \mapsto
e^{i\phi}\ket{1}$.  We also know that  the implementation of a phase operation
tend to only approximate its idealisation. For example, the implementation of
$P(\phi)$ might be $P(\phi + \epsilon)$ for some error $\epsilon$, meaning that
we will not rotate along the $z$-axis precisely $\phi$ radians but $\phi +
\epsilon$ instead. We will now compute an upper bound for the distance $\norm{
r(P(\phi + \epsilon)) - r(P(\phi))}_\diamond$ and additionally show that it
tends to $0$ when $\epsilon$ approaches $0$. The crucial observation here is
that for all quantum states $\cos \frac{\theta}{2} \ket{0} + e^{i \varphi  }
\sin \frac{\theta}{2}\ket{1} \in \C^2$ we have,
\begin{flalign*}
       &\, \norm{ 
                r(P(\phi)) \left ( \cos \frac{\theta}{2}
                \ket{0} + e^{i \varphi  } 
                \sin \frac{\theta}{2}\ket{1} \right ) - 
                r(P(\phi + \epsilon)) 
                \left ( \cos \frac{\theta}{2}
                \ket{0} + e^{i \varphi  } 
                \sin \frac{\theta}{2}\ket{1} \right )
        }_1 & \\
        &= \norm{ 
                \left ( \cos \frac{\theta}{2}
                \ket{0} + e^{i (\varphi + \phi) } 
                \sin \frac{\theta}{2}\ket{1} \right ) - 
                \left ( \cos \frac{\theta}{2}
                \ket{0} + e^{i (\varphi + \phi + \epsilon) } 
                \sin \frac{\theta}{2}\ket{1} \right )
        }_1 & \\
        &= \lvert \sin \theta \rvert 
        \norm{(\cos (\varphi + \phi) - \cos(\varphi + \phi + \epsilon),
          \sin (\varphi + \phi) - \sin(\varphi + \phi + \epsilon),
0 )}_2 & \\
        &\leq
        \norm{(\cos (\varphi + \phi) - \cos(\varphi + \phi + \epsilon),
          \sin (\varphi + \phi) - \sin(\varphi + \phi + \epsilon),
0 )}_2 &  \\
       & \leq
\norm{(K_1 \epsilon, K_2 \epsilon, 0 )}_2& \\
       & \leq
\max(K_1,K_2) \norm{(\epsilon,\epsilon,0)}_2&
\end{flalign*}
The penultimate step arises from both functions $\cos$ and $\sin$ being
Lipschitz continuous with $K_1$ and $K_2$ the corresponding Lipschitz factors.
Then by an application of~\cref{theo:smaller} we obtain,
\begin{equation}
\label{eq:dist}
\norm{ r(P(\phi + \epsilon)) - r(P(\phi))}_\diamond \leq 
\max(K_1,K_2) \norm{(\epsilon,\epsilon,0)}_2
\end{equation}
Finally the function defined by $f(\epsilon) = \max(K_1,K_2)
\norm{(\epsilon,\epsilon,0)}_2$ clearly respects $f(0) = 0$ and is continuous
which entails that the distance $\norm{ r(P(\phi + \epsilon)) -
r(P(\phi))}_\diamond$ tends to $0$ as $\epsilon$ approaches $0$.

Let us now address item~\eqref{i:ho} which concerns the fact that the category
$\Iso$ is not monoidal closed. Unfortunately, the category $\CPTP$ is not
monoidal closed either~\cite{selinger04}. We can however use general results of
category theory to overcome this issue: specifically, we can \emph{embed}
$\CPTP$ into the $\Met$-autonomous category $[\CPTP^{\cop},\Met]$ of
$\Met$-valued presheaves. Before proceeding we need to recall further
aspects of (enriched) category theory.

First it is well-known that $\CPTP$ is small and $\Met$ is (co)complete. Thus
by general results~\cite{kelly82,borceux94}, we obtain a
$\Met$-enriched category $[\CPTP^{\cop},\Met]$  whose objects are
\emph{$\Met$-enriched} functors $\CPTP^{\cop} \to \Met$ and for two such
functors $\funF$ and $\funG$ the corresponding hom-object is given by the
enriched end formula,
\[
        [\CPTP^\cop, \Met](F,G) \cong \int_{n} \Met(\funF n, \funG n)
\]
We also obtain a $\Met$-enriched Yoneda \emph{embedding} functor $\Yoneda :
\CPTP \to [\CPTP^\cop,\Met]$ that sends an object $n \in \CPTP$ to $\CPTP(-,n)
: \CPTP^\cop \to \Met$ and that sends a quantum channel $T : n \to m$ to the
morphism $(T \comp - ) : \CPTP(-,n) \to \CPTP(-,m)$. Let us analyse the
distance between two quantum channels $T$ and $S$ when in the form $\Yoneda(T)$,
$\Yoneda(S)$. To that effect we will recur to the enriched Yoneda
lemma~\cite{kelly82} which in our setting establishes an isometry,
\[
        [\CPTP^\cop,\Met](\CPTP(-,n), \funF) \cong F n
\]
for every object $n \in \CPTP$ and $\Met$-enriched functor $\funF : \CPTP^\cop
\to \Met$.  By instantiating $\funF$ with $\CPTP(-,m)$ we obtain an isometry,
\[
        [\CPTP^\cop,\Met](\CPTP(-,n), \CPTP(-,m)) \cong \CPTP(n,m)
\]
and therefore the distance between $\Yoneda(T)$ and $\Yoneda(S)$ must be equal
to that of $T$ and $S$ in $\CPTP(n,m)$. So the Yoneda embedding $\Yoneda :
\CPTP \to [\CPTP^\cop,\Met]$ indeed faithfully embeds the $\Met$-enriched
categorical structure of $\CPTP$ into $[\CPTP^\cop,\Met]$. We can actually do
better than this by recurring to Day's work on equipping functor categories
with enriched biclosed structures~\cite{day70,day70b}. Specifically, by virtue
of $\CPTP$ being small and $\Met$ being (co)complete we can equip
$[\CPTP^\cop,\Met]$ with a $\Met$-enriched symmetric monoidal structure given
by Day convolution,
\[
        F \otimes_D G \cong
        \int^{n,m} \funF n \otimes \funG m \otimes
        \CPTP(-, n \otimes m)
\]
\begin{thmC}[\cite{day70,day70b,im86}]
 The Yoneda embedding functor $\Yoneda : \CPTP \to [\CPTP^\cop,\Met]$
 is symmetric strong monoidal.
\end{thmC}
In particular, we obtain $\Yoneda(n) \otimes_D \Yoneda(m) \cong \Yoneda(n
\otimes m)$, and when $\Yoneda(n) \otimes_D \Yoneda(m)$ is seen in the form
$\Yoneda(n \otimes m)$ the swap operation corresponds to $\Yoneda(\sw) :
\Yoneda(n \otimes m) \to \Yoneda(m \otimes n)$. Finally given two
$\Met$-enriched functors $F,G : \CPTP^\cop \to \Met$ their exponential is
defined as,
\[
        F \multimap_D G \cong \int_{n} \Met(F n, G(- \otimes n))
\]
This category of enriched presheaves thus overcomes
issues \eqref{i:ho} and \eqref{i:metric} discussed above. We also obtain the
sequence of symmetric strong monoidal functors,
\[
        \xymatrix{
                \Iso \ar[r]^(0.45){r} & 
                \CPTP \ar@{^{(}->}[r]^(0.35){\Yoneda} 
                & [\CPTP^\cop, \Met]
        }
\]
We are finally ready to build a model for the metric theory of a quantum random
walk presented above. Recall that the class of ground types is
$\{\prog{qbit},\prog{pos}\}$. We interpret $\sem{\prog{qbit}}= \Yoneda r(2)$
and $\sem{\prog{pos}} = \Yoneda r(n)$ where we allow $n$ to be any  
natural number $n \geq 1$.
Then recall that the class of operation symbols is given by $\{ \prog{H} : \prog{qbit}
\to \prog{qbit},\ \prog{H^\epsilon} : \prog{qbit} \to \prog{qbit},\ \prog{S} :
\prog{qbit} \otimes \prog{pos} \to \prog{qbit} \otimes \prog{pos} \}$. We
interpret $\sem{\prog{H}} = \Yoneda r(R_y(\frac{\pi}{2}) \comp P(\phi))$ and
$\sem{\prog{H}^\epsilon} = \Yoneda r(R_y(\frac{\pi}{2}) \comp P(\phi +
\delta))$ where we allow $\delta$ to be any non-negative real number such that
$\max(K_1,K_2)\norm{(\delta,\delta,0)}_2 \leq \epsilon$.  Then recall
that for any finite set $\{ 0, \dots, n - 1\}$ we have the operations increment
$\oplus\ 1$ and decrement $\ominus\ 1$ modulo $n$. This gives rise to the
isometry $S : 2 \otimes n \to 2 \otimes n$,
\[         S(\ket{0} \otimes \ket{i}) = \ket{0} \otimes \ket{i \ominus 1}
        \hspace{1cm}
        S(\ket{1} \otimes \ket{i}) = \ket{1} \otimes \ket{i \oplus 1}
\]
(it is an isometry because it is a permutation of the canonical basis of $\C^2
\otimes \C^n$).  Intuitively, $S$ uses the left qubit as control to either move
to the left or to the right in the circle. We then interpret $\sem{\prog{S}}
= \Yoneda r(S)$. The only thing that remains to prove is that the axiom $x :
\prog{qbit} \vljud \prog{H}(x)=_\epsilon \prog{H}^\epsilon(x) : \prog{qbit}$
propounded above is sound in this interpretation. So we reason,
\begin{flalign*}
        & \, a \left (\sem{\prog{H}(x)}, 
        \sem{\prog{H^\epsilon}(x)} \right)
        & \\ 
        & =
        a \left (\Yoneda r(R_y(\textstyle{\frac{\pi}{2}}) \comp P(\phi)), 
        \Yoneda r(R_y(\textstyle{\frac{\pi}{2}}) \comp P(\phi + \delta)) \right )
        &  \\
        & =
        a \left (r(R_y(\textstyle{\frac{\pi}{2}}))\comp r(P(\phi)), 
r(R_y(\textstyle{\frac{\pi}{2}})) \comp r(P(\phi + \delta)) \right )& 
\{ \text{Enriched Yoneda}\}\\
        & \leq
a \left (r(P(\phi)), r(P(\phi + \delta)) \right )&
\{\text{$\CPTP$ is $\Met$-enriched} \}\\
        & \leq
\max(K_1,K_2)\norm{(\delta,\delta,0)}_2 & \\
        & \leq \epsilon
\end{flalign*}
We have therefore obtained a concrete higher-order model for our idealised
quantum walk, its approximation, and the fact that when $\epsilon$ tends to $0$
the distance between both walks tends to $0$ as well when bounded by a specific
number of steps.

\section{Concluding notes}
\label{sec:concl}

We end the paper with three remarks. The first 
discusses the extension of our results from the linear to the so-called
\emph{affine} setting, which renders weakening admissible -- this is 
one of the simplest natural extensions of linear calculi. The second 
establishes a functorial connection between the categorical semantics of $\V
\lambda$-calculus (presented in Section~\ref{sec:main}) and previous work on
\emph{algebraic semantics} of linear logic~\cite{paiva99}.  Finally the third
provides a brief exposition of future work.

\subsection{From linear to affine}

In this subsection we introduce an \emph{affine} variant of $\lambda$-calculus,
\ie an extension of the linear version that admits a weakening rule. We will
see that all the results that we have established thus far for linear $
\lambda$-calculus and its $\V$-equational system can be easily extended to the
affine variant. 

For the same reason as in the linear case, we assume that $\mathcal{V}$ is
integral despite not being strictly necessary for the results presented below.
The grammar of types for affine $\lambda$-calculus is defined as in the linear
case. Next, it is tempting to formalise the affine nature of the calculus
merely by adding the weakening rule,
\begin{align*}
\small{
\infer[\rulename{weakening}]{\Gamma, x:\typeA \vljud v:\typeB}{\Gamma\vljud v:\typeB}
}
\end{align*}
to those presented in \cref{fig:lang}.  This indeed allows to discard $x$ in
the construction of $v$, but it also breaks the `unique derivation' principle
that was previously discussed (recall Section~\ref{sec:back}). For example, the simple
judgement $x : \typeA, y : \typeB \vljud f(x):  \typeC$ with $f : \typeA \to
\typeC$ can now be derived in at least two different ways: 
\[
	\small{
        \infer[\rulename{weakening}]{x : \typeA, y : \typeB \vljud f (x) : \typeC}{
                \infer[\rulename{ax}]{x : \typeA \vljud f (x) : \typeC}{
                x : \typeA \vljud x : \typeA }
        }
        \hspace{4cm}
        \infer[\rulename{ax}]{x : \typeA, y : \typeB \vljud f (x) : \typeC}{
        \infer[\rulename{weakening}]{x : \typeA, y : \typeB \vljud x : \typeA}{
                x : \typeA \vljud x : \typeA }
        } 
	}        
\]
An obvious alternative is to use instead the rule,
\[
        \infer{\Gamma, x : \typeA \vljud x : \typeA}{}
\]
but unfortunately it also breaks the unique derivation principle.  This can
be seen for example with the simple judgement $x : \typeA, y : \typeB, z :
\typeC \vljud f(x,y) : \typeD$. So our approach is to add instead the term
formation rule in \cref{fig:dis} to those in \cref{fig:lang}. This rule marks
a term $v$ as \emph{discardable}, which can then be effectively discarded via
Rule $\rulename{\typeI_e}$.  We thus obtain a weakening rule,
\[
        \infer{\Gamma, x : \typeA \vljud \prog{dis}(x)\ \prog{to}\
        \ast.\ v  : \typeB}
        { \Gamma \vljud v : \typeB}
\]
\begin{figure*}[h!]
\small{
	{\renewcommand{\arraystretch}{1.1}
        \begin{tabular}{|c | c | }
	\hline 
	Term formation rule
	& 
	Semantics rule
        \\
	\hline
        \infer[\rulename{discardable}]
	{\Gamma \vljud \prog{dis}(v) : \typeI}
    {\Gamma \vljud v : \typeA} 
	&
        \infer[]{\sem{\Gamma \vljud \prog{dis}(v) : \typeI} =!_{\sem{\typeA}} \comp f}
	{\sem{\Gamma \vljud v : \typeA} = f} \\
	\hline
        \multicolumn{2}{|c|}{ Equation } \\
        \hline
        \multicolumn{2}{|c|}{
                $\begin{aligned}[t]
                  v & =  \prog{dis}(x_1)\ \prog{to}\ \ast.\ \dots\
                  \prog{dis}(x_{n-1})\ \prog{to}\ \ast.\ \prog{dis}(x_n)
          \end{aligned}$
        }
        \\
        \hline
	\end{tabular}
	}
}
\caption{Additional data for affine $\lambda$-calculus.}
 \label{fig:dis}
\end{figure*}
Substitution is defined in the expected way and like in
Section~\ref{sec:back} we obtain the following result. 
\begin{thm}\label{thm:propertiesA}
        The affine $\lambda$-calculus defined above enjoys the following properties:
        \begin{enumerate}
                \item (Unique typing) For any two judgements $\Gamma \vljud v :
                        \typeA$ and $\Gamma \vljud v : \typeA'$, we have
                        $\typeA = \typeA'$; \label{i:unique_type}
                \item (Unique derivation) Every judgement $\Gamma \vljud v :
                        \typeA$ has a unique derivation; \label{i:unique_der}
                \item (Exchange) For every judgement $\Gamma, x : \typeA, y :
                        \typeB, \Delta \vljud v : \typeC$ \label{i:exch} we can
                        derive $\Gamma, y : \typeB, x : \typeA, \Delta \vljud v
                        : \typeC$; 
                \item (Substitution) For all judgements $\Gamma, x : \typeA
                        \vljud v : \typeB$ and $\Delta \vljud w : \typeA$ we
                        can derive $\Gamma,\Delta \vljud v[w/x] : \typeB$.
                        \label{i:subst} 
	\end{enumerate}
\end{thm}

\begin{proof}
The proof is analogous to that of \cref{thm:properties}.
\end{proof}        

Our next step is to interpret the rule $\rulename{discardable}$ in a sensible
way. To that effect we move from the setting of autonomous categories to that
of \emph{semi-Cartesian autonomous categories}, or more specifically autonomous
categories whose unit object $I$ is terminal. In order to not overburden
nomenclature, let us call such categories affine. Given an affine
category $\catC$ and a $\catC$-object $X$ denote by $!_X = X \to I$ the
terminal map to the unit object. We then add to the rules of \cref{fig:lang_sem}
(which define the interpretation of linear $\lambda$-calculus on autonomous
categories) the rule for the interpretation of the $\prog{dis}$ construct, 
given in the second column of \cref{fig:dis}.  

We now focus on equipping affine $\lambda$-calculus with an equational system.
Recall  the axiomatics of autonomous categories which was provided in
\cref{fig:eqs}. We extend it with the equation in the bottom line of
\cref{fig:dis}. It can be seen as an $\eta$-equation which intuitively states
that all judgements $\Gamma \vljud v : \typeI$ (with $\Gamma = x_1 : \typeA_1,
\dots, x_n : \typeA_n$) carry no different information than that of just
discarding all variables available in context $\Gamma$.  It is clear that the
equation must be sound in affine categories, for it only involves judgements of
type $\typeI$ which is always interpreted as a terminal object.

\begin{defi}[Affine $\mathcal{V}\lambda$-theories and their
        models]\label{defn:avtheory} 
        Consider a tuple $(G,\Sigma)$ consisting of a class $G$ of ground types
        and a class of sorted operation symbols $f : \typeA_1,\dots,\typeA_n
        \to \typeA$ with $n \geq 1$. An affine $\mathcal{V}\lambda$-theory
        $((G,\Sigma),Ax)$ is a tuple such that $Ax$ is a class of
        $\mathcal{V}$-equations-in-context over \emph{affine} $\lambda$-terms
        built from $(G,\Sigma)$. 
       
        Consider an affine $\mathcal{V}\lambda$-theory $((G,\Sigma),Ax)$ and a
        $\VCatSe$-enriched \emph{affine} category $\catC$. Suppose that for
        each $X \in G$ we have an interpretation $\sem{X}$ as a $\catC$-object
        and analogously for the operation symbols. This interpretation
        structure is a model of the theory if all axioms in $Ax$ are satisfied
        by the interpretation.
\end{defi}
Let $Th(Ax)$ be the smallest class that contains $Ax$, that is closed under the
rules of Fig.~\ref{fig:eqs}, bottom line of~\cref{fig:dis}, of
\cref{fig:theo_rules}, and the compatibility rule,
\[
        \infer{\prog{dis}(v) =_q \prog{dis}(w)}{v =_q w}
\]
The elements of $Th(Ax)$ are called theorems of the
theory.

\begin{thm}[Soundness and Completeness]\label{theo:sound_compl_aff} 
        Consider an affine $\V \lambda$-theory $\mathscr{T}$.  An equation in
        context $\Gamma \vljud v =_q w : \typeA$ is a theorem of $\mathscr{T}$
        iff it is satisfied by all models of the theory.  
\end{thm}

\begin{proof}
        The proof of soundness is analogous to that of
        \cref{theo:sound_compl2}. The proof of completeness is
        also analogous to that of \cref{theo:sound_compl2}, but with
        the important difference that we need to show that the generated
        syntactic category $\Syn(\mathscr{T})$ is affine. So let us 
        consider two judgements $x : \typeA \vljud v : \typeI$ 
        and $x : \typeA \vljud w : \typeI$ and reason in the following way:
        \[
                v = \prog{dis}(x) = w 
        \]
        This entails that $[v] = [w]$ by the definition of a syntatic category.
\end{proof}        
All other results that we have established for linear $\V \lambda$-theories
(most notably, the equivalence theorem presented in \cref{theo:fequiv}) can be
obtained for the affine case as well, simply by retracing the path walked in
Section~\ref{sec:equiv}. So we omit the repetition of such steps here.

To conclude this subsection, let us briefly illustrate the affine $\V$-equational
system at work in relation to the linear case. Observe that the equations $w
=_q w'$ and $u =_r u'$ between \emph{linear} $\lambda$-terms give rise to,
\begin{align*}
  (\lambda x.\ \lambda y.\ v)\ w\ u =_{q \otimes r}
  (\lambda x.\ \lambda y.\ v)\ w'\ u'
\end{align*}
which intuitively states that the `differences' $q$ and $r$ between the input
terms $w,w'$ and $u,u'$ compound when applied to the same operation. Now let us
see what happens when the term $(\lambda x.\ \lambda y.\ v)$ does not use one
of the arguments, for instance it discards $x$ (which should not happen in the
linear case). By taking advantage of Rule \textbf{(trans)}, the fact that $k =
\top$, and of equation $\prog{dis}(w)=  \prog{dis}(x) = \prog{dis}(w')$
(discussed above), we logically deduce:
\begin{flalign*}
 & \, (\lambda x.\ \lambda y.\ \prog{dis}(x) \> \prog{to} \> \ast.\ v)\ w\ u & \\
 & = (\lambda y.\ \prog{dis}(w) \> \prog{to} \> \ast.\ v)\ u & \\
 & = (\lambda y.\ \prog{dis}(w') \> \prog{to} \> \ast.\ v)\ u & \\
 & =_{r} (\lambda y.\ \prog{dis}(w') \> \prog{to} \> \ast.\ v)\ u' & \\
 & = (\lambda x.\ \lambda y.\ \prog{dis}(x) \> \prog{to} \> \ast.\ v)\ w'\ u' &
\end{flalign*}
Note that the value $q$ is now \emph{not} accounted for in the relation between
the $\lambda$-terms $(\lambda x.\ \lambda y.\ \prog{dis}(x) \> \prog{to} \>
\ast.\ v)\ w\ u$ and $(\lambda x.\ \lambda y.\ \prog{dis}(x) \> \prog{to} \>
\ast.\ v)\ w'\ u'$; as a consequence of $w$ and $w'$ being discarded by the
operation. This is a manifestation of the available interplay between
discarding resources and $\mathcal{V}$-equations. 

\subsection{Functorial connection to previous work}
Let us introduce a simple yet instructive functorial connection between (1) the
categorical semantics of linear $\lambda$-calculus with the
$\mathcal{V}$-equational system, (2) the categorical semantics of linear
$\lambda$-calculus with the equational system of Section~\ref{sec:folin}, and (3) the
\emph{algebraic semantics} of the exponential free, multiplicative fragment of
linear logic. First we need to recall some well-known facts. As detailed
before, typical categorical models of linear $\lambda$-calculus and its
classical equational system are locally small autonomous categories, \ie
$\Set$-enriched autonomous categories. The latter form the category
$\Set$-$\Aut$ whose morphisms are autonomous functors. Then the usual algebraic
models of the exponential free, multiplicative fragment of linear logic are the
so-called \emph{lineales}~\cite{paiva99}. In a nutshell, a lineale is a poset
$(X,\leq)$ paired with a commutative monoid operation $\otimes : X \times X
\to X$ that satisfies certain conditions. Lineales are almost quantales: the
only difference is that they do not require $X$ to be cocomplete. The key idea
in algebraic semantics is that the order $\leq$ in the lineale encodes the
logic's entailment relation.

A functorial connection between $\Set$-autonomous categories and lineales (\ie
between (2) and (3)) is already stated in~\cite{paiva99} and is based on the
following two observations. First, (possibly large) lineales can be seen as
\emph{thin} autonomous categories, which are equivalently $\{0 \leq
1\}$-enriched autonomous categories. Let us use $\{0 \leq 1\}$-$\Aut$ denote
the category of $\{ 0 \leq 1\}$-enriched autonomous categories and $\{ 0
\leq 1 \}$-enriched autonomous functors.  Second, the inclusion functor $\{0
\leq 1\}$-$\Aut$ $\hookrightarrow$ $\Set$-$\Aut$ has a left adjoint which
\emph{collapses all morphisms} of a given autonomous category $\catC$
(intuitively, it eliminates the ability of $\catC$ to differentiate different
terms between two types).  This provides an adjoint situation between (2) and
(3). 

We can then expand the connection just described to our categorical semantics of
linear $\lambda$-calculus and corresponding $\mathcal{V}$-equational system
(\ie (1)) in the following way:  the forgetful functor $\VCat \to \Set$ has a
left adjoint $D : \Set \to \VCat$ which sends a set $X$ to $DX = (X,d)$, where
$d(x_1,x_2) = k$ if $x_1 = x _2$ and $d(x_1,x_2) = \bot$ otherwise.  This left
adjoint is (strict) monoidal, specifically we have $D(X_1 \times X_2) = DX_1
\otimes DX_2$ and $I = (1,(\ast,\ast) \mapsto k) = D1$. This gives rise to
the change-of-base functors on the left side of the diagram,
\[
\begin{tikzcd}
  \VCat\text{-} \Aut \arrow[r, shift right=1ex]
  & \Set\text{-}\Aut \arrow[l, "\hat D"', shift right=1ex]
  \arrow[phantom, l, "\scriptscriptstyle{\bot}"]
  \arrow[r, "C", shift left=1ex]
  & \{0 \leq 1\}\text{-}\Aut
  \arrow[phantom, l, "\scriptscriptstyle{\bot}"]
  \arrow[l, shift left=1ex]
\end{tikzcd}
\]
The functor $\hat D$ equips the
hom-sets of a $\Set$-autonomous category with the corresponding discrete
$\mathcal{V}$-category and $C$ collapses all morphisms of a $\Set$-autonomous
category as described earlier. The right adjoint of $\hat D$ forgets the
$\mathcal{V}$-categorical structure between terms (\ie between morphisms) and the
right adjoint of $C$ is the inclusion functor mentioned earlier. Note that
$\hat D$ restricts to $\VCatSe$-$\Aut$ and $\VCatSS$-$\Aut$, and thus we
obtain a functorial connection between the categorical semantics of linear
$\lambda$-calculus with the $\mathcal{V}$-equational system (\ie (1)), (2),
and (3). In essence, the connection formalises the fact that our categorical
models admit a richer structure between terms (\ie morphisms) than the
categorical models of linear $\lambda$-calculus and its classical equational
system. The latter in turn permits the existence of different terms between two
types as opposed to the algebraic semantics of the exponential free,
multiplicative fragment of linear logic. The connection also shows that models
for (2) and (3) can be mapped into models of our categorical semantics by
equipping the respective hom-sets with a trivial, discrete structure.

\subsection{Future work}

We introduced the notion of a $\mathcal{V}$-equation which generalises the
notions of  equation, inequation~\cite{kurz2017quasivarieties,adamek20}, and
metric equation~\cite{mardare16,mardare17}. We then presented a sound and
complete $\mathcal{V}$-equational system, illustrated with different examples
concerning real-time, probabilistic, and quantum computing. We additionally
showed that linear $\V \lambda$-theories are the syntactic counterpart of
$\VCatSe$-enriched autonomous categories, a connection that allows to
seamlessly invoke both logical and categorical constructs when studying any of
these structures.

Linear $\lambda$-calculus is at the root of different ramifications of
$\lambda$-calculus that relax resource-based conditions in different ways.
We already presented here the affine case which allows to discard resources. We
are now studying the possibility of extending our results to \emph{mixed
linear-non-linear} calculus~\cite{benton1994mixed}. Actually we already have some
results in this regard, specifically for the case in which linear
$\lambda$-calculus is extended with a \emph{graded exponential
modality}~\cite{neves23}. The general idea of the \emph{op. cit.} is to extend
the notion of a model studied in this paper (\ie\ $\VCat$-enriched
autonomous category) with that of a graded comonad that obeys a lax notion of
$\VCat$-enrichment -- the latter can be seen as a generalised form of Lipschitz
continuity.

Next, our main examples of $\mathcal{V}\lambda$-theories (see
Section~\ref{sec:examples}) used either the Boolean or the metric quantale.  We would
like to study linear $\mathcal{V}\lambda$-theories whose underlying quantales
are neither the Boolean nor the metric one, for example the ultrametric
quantale which is (tacitly) used to interpret Nakano's guarded
$\lambda$-calculus \cite{birkedal10} and also to interpret a higher-order
language for functional reactive programming~\cite{krishnaswami11}. Another
interesting quantale is the G\"{o}del one which is a basis for fuzzy
logic~\cite{denecke13} and whose $\mathcal{V}$-equations give rise to what we
call fuzzy inequations.

Finally we plan to further explore the connections between our work and
different results on metric universal
algebra~\cite{mardare16,mardare17,rosicky20} and inequational universal
algebra~\cite{kurz2017quasivarieties,adamek20,rosicky20}. For example, an
interesting connection is that the monad construction presented
in~\cite{mardare16} crucially relies on quotienting  a pseudometric space into
a metric space -- this is a particular case of quotienting a
$\mathcal{V}$-category into a separated $\mathcal{V}$-category which we also
use in our work.

\section*{Acknowledgments}
This work is financed by National Funds through FCT - Fundação para a Ciência e
a Tecnologia, I.P. (Portuguese Foundation for Science and Technology) within
project IBEX, with reference PTDC/CCI-COM/4280/2021.  We are also thankful for
the reviewers' helpful feedback. 
\bibliographystyle{alphaurl}
\bibliography{biblio}

\appendix
  
\end{document}

%% file: macros.tex
\newcommand{\bern}{\mathtt{bernoulli}}

\DeclareMathOperator{\un}{\mathsf{u}}

\newcommand{\Cats}{\catfont{Cat}}

\newcommand{\Met}{\catfont{Met}}

\newcommand{\Free}[1]{\sfunfont{F}}

\newcommand{\typeI}{\typefont{I}}

\newcommand{\Shuff}{\mathrm{Sf}}
\newcommand{\VCat}{\mathcal{V}\text{-}\Cats}
\newcommand{\VCatSy}{\mathcal{V}\text{-}\Cats_{\mathsf{sym}}}
\newcommand{\VCatSe}{\mathcal{V}\text{-}\Cats_{\mathsf{sep}}}
\newcommand{\VCatSS}{\mathcal{V}\text{-}\Cats_{\mathsf{sym,sep}}}
\newcommand{\Lang}{\mathrm{Lng}}
\newcommand{\Syn}{\mathrm{Syn}}
\DeclarePairedDelimiter\abs{\lvert}{\rvert}%
\DeclarePairedDelimiter\norm{\lVert}{\rVert}%
\makeatletter
\let\oldabs\abs
\def\abs{\@ifstar{\oldabs}{\oldabs*}}
\let\oldnorm\norm
\def\norm{\@ifstar{\oldnorm}{\oldnorm*}}
\makeatother

\DeclarePairedDelimiter\ket{\lvert}{\rangle}

\newcommand{\catfont}[1]{\mathsf{#1}}
\newcommand{\cop}{\catfont{op}}

\newcommand{\catC}{\catfont{C}}
\newcommand{\catD}{\catfont{D}}
\newcommand{\catA}{\catfont{A}}

\newcommand{\Set}{\catfont{Set}}
\newcommand{\Pos}{\catfont{Pos}}

\newcommand{\Subs}[2]{\catfont{Sub}_{}}

\newcommand{\Aut}{\catfont{Aut}}
\newcommand{\funfont}[1]{#1}
\newcommand{\funF}{\funfont{F}}

\newcommand{\funG}{\funfont{G}}

\newcommand{\sfunfont}[1]{\mathrm{#1}}

\newcommand{\Id}{\sfunfont{Id}}

\newcommand{\Yoneda}{\sfunfont{Y}}

\newcommand\adjunct[2]{\xymatrix@=8ex{\ar@{}[r]|{\top}\ar@<1mm>@/^2mm/[r]^{{#2}}
& \ar@<1mm>@/^2mm/[l]^{{#1}}}}
\newcommand\adjunctop[2]{\xymatrix@=8ex{\ar@{}[r]|{\bot}\ar@<1mm>@/^2mm/[r]^{{#2}}
& \ar@<1mm>@/^2mm/[l]^{{#1}}}}
\newcommand\retract[2]{\xymatrix@=8ex{\ar@{}[r]|{}\ar@<1mm>@/^2mm/@{^{(}->}[r]^{{#2}}
& \ar@<1mm>@/^2mm/@{->>}[l]^{{#1}}}}

\newcommand{\sem}[1]{\llbracket #1 \rrbracket}

\newcommand{\app}{\mathrm{app}}

\newcommand{\comp}{\cdot}

\DeclareMathOperator{\id}{\mathsf{id}}
\DeclareMathOperator{\sw}{\mathsf{sw}}
\DeclareMathOperator{\spl}{\mathsf{sp}}
\DeclareMathOperator{\sh}{\mathsf{sh}}
\DeclareMathOperator{\join}{\mathsf{jn}}
\DeclareMathOperator{\exch}{\mathsf{exch}}

\newcommand{\defeq}{\triangleq}
\newcommand{\resp}{resp.\xspace~}

\newcommand{\ie}{\textit{i.e.}\xspace~}
\newcommand{\eg}{\textit{e.g.}\xspace~}

\newcommand{\Nats}{\mathbb{N}}
\newcommand{\Reals}{\mathbb{R}}

\newcommand{\C}{\mathbb{C}}
\newcommand{\V}{\mathcal{V}}

\newcommand{\typefont}[1]{\mathbb{#1}}

\newcommand{\typeA}{\typefont{A}}
\newcommand{\typeB}{\typefont{B}}
\newcommand{\typeC}{\typefont{C}}

\newcommand{\typeD}{\typefont{D}}
\newcommand{\rulename}[1]{(\mathrm{\mathbf{#1}})}
\newcommand{\vljud}{\rhd}

\newcommand{\closValBP}[2]{\mathsf{Values}(#1, #2)}
\newcommand{\prog}[1]{\mathtt{#1}}